\documentclass[MSNbibl,number,citesort,seceqn,dvips]{arxbj}
\usepackage{graphicx}
\aid{0}
\volume{19}
\issue{5B}
\pubyear{2013}
\firstpage{2330}
\lastpage{2358}
\doi{10.3150/12-BEJ454} 

\makeatletter
\newcommand{\pa}{\operatorname{{pa}}}
\newcommand{\nei}{\operatorname{{ne}}}
\newcommand{\spo}{\operatorname{{sp}}}
\newcommand{\an}{\operatorname{{an}}}
\newcommand{\dse}{\perp}

\newcommand{\cd}{|}

\newcommand{\nl}{\\}

\newcommand{\margn}{\mbox{\protect\raisebox{-1pt}{\mbox{\protect
\includegraphics{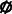}
}}}}
\newcommand{\condnc}{\mbox{\protect\raisebox{-1pt}{\mbox{\protect
\includegraphics{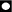}
}}}}
\newcommand{\squarea}{\mbox{\protect\raisebox{-1pt}{\mbox{\protect
\includegraphics{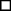}
}}}}

\newcommand{\nn}[0]{\hspace*{.7em}}
\newcommand{\ful}{\,\mbox{$\frac{ \nn\nn\;}{ \nn\nn
}$}\,}
\newcommand{\fla}{\,\mbox{$ \prec
\!\!\!\!\!\frac{\nn\nn}{\nn}$}\,}
\newcommand{\fra}{\,\mbox{$ \frac{\nn
\nn}{\nn
}\!\!\!\!\! \succ\! \hspace{.25ex}$}\,}
\newcommand{\arc}{\,\mbox{$  \prec
\!\!\!\!\!\frac{\nn\nn}{\nn}
\!\!\!\!\!
\succ\! \hspace{.25ex}$}\,}

\newtheorem{prop}{Proposition}
\newtheorem{coro}{Corollary}
\newtheorem{lemma}{Lemma}
\newproclaim{alg}{Algorithm}
\newtheorem{theorem}{Theorem}

\renewcommand{\citep}{\cite}
\renewcommand{\epsilon}{\varepsilon}
\newremark{example}{Example}
\makeatother

\begin{document}
\begin{frontmatter}

\title{Stable mixed graphs}
\runtitle{Stable mixed graphs}

\begin{aug}
\author{\fnms{Kayvan} \snm{Sadeghi}\corref{}\ead[label=e1]{sadeghi@stats.ox.ac.uk}}
\runauthor{K. Sadeghi} 
\address{Department of Statistics, University of Oxford, 1 South Parks
Road, Oxford, OX1 3TG, United Kingdom. \printead{e1}}
\end{aug}

\received{\smonth{10} \syear{2011}}
\revised{\smonth{5} \syear{2012}}

%
\begin{abstract}
In this paper, we study classes of graphs with three types of edges
that capture the modified independence structure of a directed acyclic
graph (DAG) after marginalisation over unobserved variables and
conditioning on selection variables using the $m$-separation criterion.
These include MC, summary, and
ancestral graphs. As a modification of MC graphs, we define the class
of ribbonless graphs (RGs) that permits the use of the $m$-separation
criterion. RGs contain summary and ancestral graphs as subclasses, and
each RG can be generated by a DAG after marginalisation and
conditioning. We derive simple algorithms to generate RGs, from given
DAGs or RGs, and also to generate summary and ancestral graphs in a
simple way by further extension of the RG-generating algorithm. This
enables us to develop a parallel theory on these three classes and to
study the relationships between them as well as the use of each class.
\end{abstract}

%

\end{frontmatter}

\section{Introduction}\label{sec1}
\noindent\emph{Introduction and motivation.} In graphical Markov models,
graphs have been used to represent conditional independence
statements of sets of random variables. Nodes of the graph
correspond to random variables and edges typically capture
dependencies. Different classes of graphs with different interpretation
of independencies have been defined and studied in the literature.

One of the most important classes of graphs in graphical models is the
class of directed acyclic graphs (DAGs) \citep{kii84,lau96}. Their
corresponding Markov models, often known under the name of Bayesian networks
\citep{pea88}, have direct applications to a
wide range of areas including econometrics, social sciences, and
artificial intelligence. When, however, some variables are unobserved,
that is
also called latent or hidden, one can in general no longer capture the
implied independence model among observed variables by a DAG. In this
sense, the DAG models
are not stable under marginalisation. A
similar problem occurs because DAG models are not
stable under conditioning \citep{wer94,pea94,cox96}.

This makes it necessary to identify and study a class of graphs that
includes DAGs and is stable under marginalisation and conditioning in
the sense that it is able to express the induced independence model
after marginalisation and conditioning through an object of the same class.
The methods that have been used to solve this problem employ three different
types of edges instead of a single type.

Three known classes of graphs
have previously been suggested for this purpose in the literature. We
specifically call these stable mixed graphs (under marginalisation and
conditioning) and they include MC graphs (MCGs) \citep{kos02}, summary
graphs (SGs) \citep{wer94,wer08}, and ancestral
graphs (AGs) \citep{ric02}.

MCGs do not use the same interpretation of independencies, called the
$m$-separation criterion, as the other types of stable mixed graphs. In
this paper, we use similar methods as in \citep{kos02} to derive a
modification of the class of MCGs to use $m$-separation, which we call
ribbonless graphs (RGs). The class of RGs is exactly the class with
three types of edges that is generated after marginalisation over and
conditioning on the node sets of a DAG. More importantly, we extend the
RG-generating algorithm to generate summary and ancestral graphs in a
theoretically neat way. These algorithms are computationally
polynomial, even though we shall not go through their computational
complexity in this paper. Defining these algorithms leads to
establishing a parallel theory for the different classes, and studying
the similarities, differences, and relationships among
them.\looseness=1

\noindent\emph{Structure of the paper.} In the next section, we
define some basic concepts of graph theory and independence models
needed in this paper.

In Section~\ref{sec3}, we define the class of RGs, give some basic
graph-theoretical definitions for these, and define the $m$-separation
criterion for interpretation of the independence structure on them.

In Section~\ref{sec4}, we formally define marginalisation and conditioning for
independence models in such a way that it conforms with marginalisation
and conditioning for probability distributions. We also formally define
stable classes of graphs.

Each of the next three sections of this paper deals with one type of
stable mixed graphs. We discuss
RGs in Section~\ref{sec5}, SGs in Section~\ref{sec6}, and AGs in Section~\ref{sec7}. In each
section, we introduce a straightforward algorithm to generate the
stable mixed graph from DAGs or from graphs of the same type. For each
type of stable mixed graph, we prove that the graphs and algorithms are
well-defined in the sense that instead of marginalising over or
conditioning on a set of nodes, by splitting the marginalisation or
conditioning set into two subsets, one can marginalise over or
condition on the first subset first, and then marginalise over or
condition on the second subset and obtain the same graph. We also prove
that the generated graphs induce the modified independence model after
marginalisation and conditioning, meaning that the generated classes
are stable under marginalisation and conditioning.

In Section~\ref{sec8}, we scrutinise the relationships between the three types
of stable mixed graphs. In Section~\ref{sec9}, we provide a discussion on the
use of the different classes of stable mixed graphs.

In the \hyperref[appm]{Appendix}, we provide the proof of lemmas, propositions, and
theorems given in the previous sections.
\section{Basic definitions and concepts}\label{sec2}

\noindent\emph{Independence models and graphs.} An \emph{independence
model} $\mathcal{J}$ over a set $V$ is a set of triples $\langle
X,Y\cd Z\rangle$ (called \emph{independence statements}), where $X$,
$Y$, and $Z$ are disjoint subsets of $V$ and $Z$
can be empty, and $\langle\varnothing,Y\cd Z\rangle$ and $\langle
X,\varnothing\cd Z\rangle$ are always included in $\mathcal{J}$. The
independence statement $\langle X,Y\cd Z\rangle$ is interpreted as
``$X$ is independent of $Y$ given $Z$''. Notice that independence
models contain probabilistic independence models as a special case. For
further discussion on independence models, see \citep{stu05}.

A \emph{graph} $G$ is a triple consisting of a \emph{node} set or
\emph{vertex} set $V$, an \emph{edge} set $E$, and a relation that with
each edge associates two nodes (not necessarily distinct), called
its \emph{endpoints}. When nodes $i$ and $j$ are the endpoints of an
edge, these are
\emph{adjacent} and we write $i\sim j$. We say the edge is \emph
{between} its two
endpoints. We usually refer to a graph as an ordered
pair $G=(V,E)$. Graphs $G_1=(V_1,E_1)$ and $G_2=(V_2,E_2)$ are called
\emph{equal} if $(V_1,E_1)=(V_2,E_2)$. In this case we write $G_1=G_2$.

Notice that the graphs that we use in this paper (and in general in the
context of graphical models) are so-called \emph{labeled graphs}, that
is, every node is considered a different object. Hence, for example,
graph $i\ful j\ful k$ is not equal to $j\ful i\ful k$.

We use the notation $\mathcal{J}^G$ for an independence model defined
over the node set of $G$. Among the independence models over the node
set $V$ of a graph $G$, those that are of interest to us \emph
{conform} with $G$, meaning that $i\sim j$ in $G$ implies $\langle
i,j\cd C\rangle\notin\mathcal{J}$ for any $C\subseteq V\setminus\{
i,j\}$. Henceforth, we assume that independence models $\mathcal{J}^G$
conform with $G$, unless otherwise stated. Notice that henceforth we
use the notation $i$ instead of $\{i\}$ for a subset consisting of a
single element $i$ in an independence statement.

\noindent\emph{Basic graph theoretical definitions.} Here we introduce some
basic graph theoretical definitions. A \emph{loop} is an edge with
the same endpoints. \emph{Multiple edges} are edges with the
same pair of endpoints. A \emph{simple graph} has neither
loops nor multiple edges.

If a graph assigns an ordered pair of nodes to each edge, then the
graph is a \emph{directed graph}.
We say that the edge is \emph{from} the first node of the ordered pair
\emph{to} the second one.
We use an arrow, $j\fra i$, to draw an edge in a directed graph.
We also call node $j$ a \emph{parent} of $i$, node $i$ a \emph{child}
of $j$ and we use the notation $\pa(i)$ for the set of all parents of
$i$ in the graph.

A \emph{walk} is a list $\langle v_0,e_1,v_1,\ldots,e_k,v_k\rangle$
of nodes and edges such that for $1\leq i\leq k$, the edge $e_i$ has
endpoints $v_{i-1}$ and $v_i$. A \emph{path} is a walk with no
repeated node or edge. A~\emph{cycle} is a walk with no repeated node
or edge except $v_0=v_k$. If the graph is simple, then a path or a
cycle can be determined uniquely by an ordered sequence of node sets.
Throughout this paper, however, we use node sequences for describing
paths and cycles even in graphs with multiple edges, but we suppose
that the edges of the path are all determined. Usually it is apparent
from the context or the type of the path which edge belongs to the path
in multiple edges.
We say a path is \emph{between} the first and the last nodes of the
list in $G$. We
call the first and the last nodes \emph{endpoints} of the path and all
other nodes \emph{inner nodes}.

A path (or a cycle) in a directed graph is \emph{direction
preserving} if all its arrows point to one direction ($\circ\fra\circ
\fra\cdots \fra\circ$). A directed graph is \emph{acyclic} if it has
no direction-preserving cycle.

If in a direction-preserving path an arrow starts at a node $j$ and an
arrow points
to a node $i$, then $j$ is an \emph{ancestor} of $i$, and $i$ a
\emph{descendant} of $j$. We use the notation $\an(i)$ for the set of
all ancestors of $i$.
\section{Independence model for ribbonless graphs}\label{sec3}

\noindent\emph{Loopless mixed graphs.} Graphs that will be discussed in
this paper are subclasses of loopless mixed graphs. A \emph{mixed
graph} is a graph containing three types of edges denoted by
arrows, arcs (two-headed arrows), and lines (solid lines). Mixed
graphs may have multiple edges of different types but do
not have multiple edges of the same type.
We do not distinguish between $i\ful j$ and $j\ful i$ or $i\arc j$
and $j \arc i$, but we
do distinguish between $j\fra i$ and $i\fra j$. Thus there are up to
four edges as a multiple edge between any two nodes. A \emph{loopless
mixed graph} (LMG) is a mixed graph that does not contain any loops (a
loop may be formed by a line, arrow, or arc).

\noindent\emph{Some definitions for mixed graphs.} For a mixed graph $H$,
we keep the same terminology introduced before for directed and
undirected graphs. We say that $i$ is a
\emph{neighbour} of $j$ if these are endpoints of a line, and $i$ is a
parent of $j$ if there is an arrow from $i$ to $j$. We also define that
$i$ is a \emph{spouse} of $j$ if these are endpoints of an arc. We use
the notations $\nei(j)$, $\pa(j)$, and $\spo(j)$ for the set of all
neighbours, parents, and spouses of $j$, respectively.

In the cases of $i\fra j$ or
$i\arc j$, we say that there is an arrowhead pointing to (at) $j$. A
path $\langle j=i_0,i_1,\ldots,i_n=i\rangle$ is \emph{from} $j$ \emph
{to} $i$ if $ji_1$ is either a line or an arrow
from $j$ to $i_1$, and $i_{n-1}i$ is either an arc or an arrow from
$i_{n-1}$ to $i$.

A \emph{{\textsf V}-configuration} or simply \emph{{\textsf V}s} is a path
with three nodes and two edges. In a mixed graph, the inner node of
three {\textsf V}s $i\fra
t\fla j$, $i\arc t\fla j$ and $i\arc t\arc j$ is
a \emph{collider} and the inner node of all other {\textsf V}s
is a \emph{non-collider} node in the {\textsf V} or more generally in a
path on which the {\textsf V} lies. We also call a {\textsf V} with collider or
non-collider inner node a \emph{collider} or \emph{non-collider {\textsf
V}}, respectively. We
may mention that a node is collider or non-collider without
mentioning the {\textsf V} or path when this is apparent from the context.
Notice that originally \citep{kii84} and in most texts, the
endpoints of a {\textsf V} are not adjacent whereas we do not use this restriction.

Two paths (including {\textsf V}s or edges) are called \emph
{endpoint-identical} if presence or lack of arrowheads pointing to
endpoints on the path are the same in both. For example, $i\fra j$,
$i\ful k\arc j$, and $i\fra k\fla l \arc j$ are all endpoint-identical
as there is an arrowhead pointing to $j$ but there is no arrowhead
pointing to $i$ on the paths.

\noindent\emph{Ribbonless graphs and its subclasses.} The largest subclass
of LMGs studied in this paper is the class of ribbonless graphs.

A \emph{ribbon} is a collider {\textsf V} $\langle h,i,j\rangle$ such that
\begin{enumerate}
\item there is no endpoint-identical edge between $h$ and $j$, that
is, there is no $hj$-arc in the case of $h\arc i\arc j$; there is no
$hj$-line in the case of $h\fra i\fla j$; and there is no arrow from
$h$ to $j$ in the case of $h\fra i\arc j$;
\item$i$ or a descendant of $i$ is an endpoint of a line or on a
direction-preserving cycle.
\end{enumerate}
A \emph{ribbonless graph} (RG) is an LMG that does not contain ribbons
as induced subgraphs.\eject

Figure~\ref{fig:3ex2} illustrates ribbons $\langle h,i,j\rangle$.
Figure~\ref{fig:3ex2nn}(a) illustrates a graph containing a
ribbon $\langle h,i,j\rangle$. Figure~\ref{fig:3ex2nn}(b) illustrates
a ribbonless graph. Notice that $\langle h,i,j\rangle$ is not here a
ribbon since there is a line between $h$ and $j$.

\begin{figure}

\includegraphics{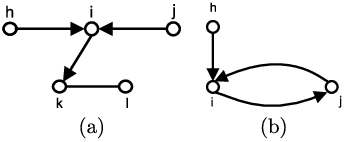}

\caption{(a) A ribbon $\langle h,i,j\rangle$, where $\nei
(i)=\varnothing$. (b) A ribbon $\langle h,i,j\rangle$.}
\label{fig:3ex2}
\end{figure}
%
\begin{figure}

\includegraphics{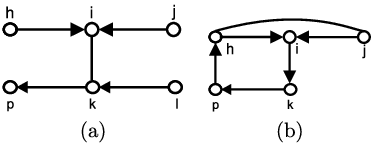}

\caption{(a) A graph that is not ribbonless. (b) A
ribbonless graph.}
\label{fig:3ex2nn}
\end{figure}

The three classes of undirected graphs (UGs) (used for concentration
graph models), bidirected graphs (BGs) (used for covariance graph
models), and DAGs are subclasses of RGs. SGs and AGs, which are studied
in this paper, are also subclasses of RGs. We use the notations
$\mathcal{RG}$, $\mathcal{SG}$, $\mathcal{AG}$, $\mathcal{UG}$,
$\mathcal{BG}$ and $\mathcal{DAG}$ for the set of all RGs, SGs, AGs,
UGs, BGs and DAGs, respectively. The common feature of all these graphs
is that these all entail independence models using the same so-called
separation criterion, which is called $m$-separation and will be
shortly defined.

\noindent\emph{The $m$-separation criterion for RGs.} The following
definition was given in \citep{ric02}.

Let $C$ be a subset of the node set $V$ of an RG. A path is
$m$-connecting given $M$ and $C$ if all its
collider nodes are in $C\cup\an(C)$ and all its non-collider nodes
are in $M$. For two other disjoint subsets of the node set $A$ and $B$
such that $M=V\setminus A\cup B\cup C$, we may just call the path
$m$-connecting given $C$ between $A$ and $B$. We say $A\dse_mB\cd C$
if there is no $m$-connecting path between $A$ and $B$ given $C$.

Notice that the $m$-separation criterion induces an independence model
$\mathcal{J}_m(G)$ on $G$ by $A\dse_m B\cd C \iff\langle A,B\cd
C\rangle\in\mathcal{J}_m(G)$.
\section{Marginalisation, conditioning and stability}\label{sec4}

\noindent\emph{Marginal and conditional independence models.} Consider an
independence model $\mathcal{J}$ over a set $V$. For $M$ a subset of
$V$, the \emph{independence model $\mathcal{J}$ after marginalisation
over $M$}, denoted by $\alpha(\mathcal{J};M,\varnothing$), is the
subset of $\mathcal{J}$ whose triples do not contain\vadjust{\goodbreak} members of $M$,
that is,
\begin{displaymath}
\alpha(\mathcal{J};M,\varnothing)= \bigl\{\langle A,B\cd D\rangle\in \mathcal{J}
\dvt(A\cup B\cup D)\cap M=\varnothing \bigr\}.
\end{displaymath}

One can observe that $\alpha(\mathcal{J};M,\varnothing)$ is an
independence model over $V\setminus M$.

For a subset $C$ of $V$, the \emph{independence model after
conditioning on $C$}, denoted by $\alpha(\mathcal{J};\varnothing
,C$), is
\begin{displaymath}
\alpha(\mathcal{J};\varnothing,C)= \bigl\{\langle A,B\cd D\rangle\dvt\langle A,B
\cd D\cup C\rangle\in\mathcal{J}\mbox{ and }(A\cup B\cup D)\cap C=\varnothing
\bigr\}.
\end{displaymath}

One can also observe that $\alpha(\mathcal{J};\varnothing,C)$ is an
independence model over $V\setminus C$.

Combining these definitions, for disjoint subsets $M$ and $C$ of $V$,
the \emph{independence model after marginalisation over $M$ and
conditioning on $C$} is
\begin{displaymath}
\alpha(\mathcal{J};M,C)= \bigl\{\langle A,B\cd D\rangle\dvt\langle A,B\cd D\cup
C \rangle\in\mathcal{J}\mbox{ and }(A\cup B\cup D)\cap(M\cup C)=\varnothing \bigr
\},
\end{displaymath}
which is an independence model over $V\setminus(M\cup C)$.

Notice here that $\alpha$ is a function from the set of independence
models and two of their subsets to the set of independence models.
Notice also that operations for marginalisation and conditioning commute.

%
\begin{figure}[b]

\includegraphics{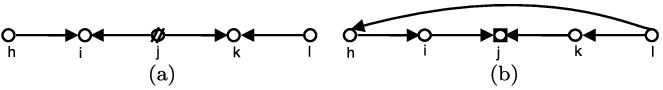}

\caption{(a) A directed acyclic graph $G_1$, by which it can
be shown that the class of DAGs is not stable under marginalisation.
($\margn\in M$.) (b) A directed acyclic graph $G_2$, by which it can
be shown that the class of DAGs is not stable under conditioning.
($\condnc\in C$.)}
\label{fig:2ex000}
\end{figure}

Marginalisation and conditioning in probability conform with
marginalisation and conditioning for independence models. Consider a
set $N=V\setminus(M\cup C)$ and a collection of random variables
$(X_\alpha)_{\alpha\in N}$ with joint density $f_{(V\setminus M)\cd
C}$. We associate an independence model to this density. It can be
shown that if $\mathcal{J}$ is the associated independence model to
the collection of random variables $(X_\alpha)_{\alpha\in V}$ with
joint density $f_V$ then the associated independence model to
$(X_\alpha)_{\alpha\in N}$ with joint density $f_{(V\setminus M)\cd
C}$ is $\alpha(\mathcal{J},M,C)$.

\noindent\emph{Stability under marginalisation and conditioning for RGs
and its subclasses.} Consider a family of graphs $\mathcal{T}$. If,
for every graph $G=(V,E)\in\mathcal{T}$ and every disjoint subsets
$M$ and $C$ of $V$, there is a graph $H\in\mathcal{T}$ such that
$\mathcal{J}_m(H)=\alpha(\mathcal{J}_m(G);\varnothing,C)$ then
$\mathcal{T}$ is stable under conditioning, and if there is a graph
$H\in\mathcal{T}$ such that $\mathcal{J}_m(H)=\alpha(\mathcal
{J}_m(G);M,\varnothing)$ then $\mathcal{T}$ is stable under
marginalisation. We call $\mathcal{T}$ \emph{stable} (under
marginalisation and conditioning) if there is a graph $H\in\mathcal
{T}$ such that $\mathcal{J}_m(H)=\alpha(\mathcal{J}_m(G);M,C)$.

Notice that if the node set of such a graph $H$ is $N$ then
$N=V\setminus(M\cup C)$.

We shall see that RGs, SGs, AGs, UGs, and BGs are stable. On the other
hand, the class of DAGs is not stable. It can be shown that $G_1$ in
Figure~\ref{fig:2ex000} is a DAG whose induced marginal independence
model cannot be represented by a DAG and $G_2$ is a DAG whose induced
conditional independence model cannot be represented by a DAG. We leave
the details as an exercise to the readers.\vadjust{\goodbreak}

\noindent\emph{Stable mixed graphs.} As the class of DAGs is not stable,
we look for stable classes of graphs that include the class of DAGs as
a subclass. In this paper, we discuss three such types of graphs,
namely RGs (as a modification of MCGs), SGs,
and AGs, and specifically call these \emph{stable mixed graphs}. We
will see that in these graphs arcs are related to marginalisation and
lines are related to conditioning.

For the graph $G_2\in\mathcal{T}$ for which $\mathcal
{J}^{G_2}=\alpha(\mathcal{J}_m(G_1);M,C)$, we use the notation
$G_2=\alpha_\mathcal{T}(G_1;M,C)$. For each type of stable mixed
graphs, we later precisely define $\alpha_\mathcal{T}$ with specific
algorithms. We call $\alpha_\mathcal{T}$ a \emph{generating
function} or more specifically a \emph{$\mathcal{T}$-generating function}.
\section{Ribbonless graphs}\label{sec5}

\noindent\emph{MC graphs and ribbonless graphs.} \emph{MCGs} only contain
the three desired types of edges. However, these are not loopless and,
in addition, in MCGs a different separation criterion is used for
inducing the independence model.
However, from an MCG
that can be generated by marginalisation and conditioning over DAGs and by a minor
modification one can generate an RG that induces the same independence model. This
modification includes adding edges between pairs of nodes connected by a ribbon such
that the generated edges preserve the arrowheads at the endpoints of the ribbon, and
removing all the loops. We shall not go through the details of this modification in this
paper, but refer readers to \cite{sad12}.
\subsection{Generating ribbonless graphs}\label{sec5.1}

\noindent\emph{A local algorithm to generate RGs from RGs.} Here we
present an algorithm to generate an RG from a given RG and two subsets
of its node set that will be marginalised over and conditioned on. This
algorithm is local in the sense that, after determining the ancestor
set of the conditioning set, it looks solely for all {\textsf V}s in the
graph and not for longer paths. Later in this section, we will show
that a graph generated by the algorithm is an RG and it induces the
marginal and conditional independence model of the input graph by using
$m$-separation.

%
\begin{table*}
\tabcolsep=0pt
\tablewidth=180pt
\caption{Types of edge induced by {\textsf V}s with inner node in
$m\in M$ or $s\in C\cup\an(C)$}\label{tab:21}
\begin{tabular*}{180pt}{@{\extracolsep{\fill}}lccc@{}}
\hline
\hphantom{0}1 & $i$\fla $m$\fla  $j$ & generates & $i$\fla $j$\nl
\hphantom{0}2 & $i$\fla $m$\ful  $j$ & generates & $i$\fla $j$\nl
\hphantom{0}3 & $i$\arc $m$\ful  $j$ & generates & $i$\fla $j$\nl
\hphantom{0}4 & $i$\fla $m$\fra  $j$ & generates & $i$\arc $j$\nl
\hphantom{0}5 & $i$\fla $m$\arc  $j$ & generates & $i$\arc $j$\nl
\hphantom{0}6 & $i$\ful $m$\fla  $j$ & generates & $i$\ful $j$\nl
\hphantom{0}7 & $i$\ful $m$\ful  $j$ & generates & $i$\ful $j$\nl
[6pt]
\hphantom{0}8 & $i$\arc $s$\fla  $j$ & generates & $i$\fla $j$\nl
\hphantom{0}9 & $i$\arc $s$\arc  $j$ & generates & $i$\arc $j$\nl
10 & $i$\fra $s$\fla  $j$ & generates & $i$\ful $j$\\
\hline
\end{tabular*}
\end{table*}

Suppose that $H$ is an RG and consider $M$ and $C$ two
disjoint subsets of the node set. There are $10$ possible
non-isomorphic {\textsf V}s in an RG, displayed in Table~\ref{tab:21}.
Notice that this table generates endpoint-identical edges to the given
{\textsf V}s. We now define the following algorithm, derived from \citep
{wer94} and \citep{kos02}. See also the appendix of \citep{wer08}.
%
\begin{alg}\label{alg:22}
$\alpha_{\mathrm{RG}}(H;M,C)$ (Generating an RG
from a ribbonless graph $H$):

Start from $H$.

Generate an endpoint identical edge between the endpoints of
collider {\textsf V}s with inner node in $C\cup\an(C)$ and non-collider
{\textsf V}s with inner node in $M$, that is, generate an appropriate edge
as in Table~\ref{tab:21} between the endpoints of every {\textsf V}
with inner node in $M$ or $C\cup\an(C)$ if the edge of the same type
does not already exist.

Apply the previous step until no other edge can be generated. Then
remove all nodes in $M\cup C$.\vadjust{\goodbreak}
\end{alg}

This method is a generalisation of the method used by \citep{lau90},
called \emph{moralisation}, as a separation criterion on DAGs. Notice
that the order of applying steps of Table~\ref{tab:21} in Algorithm \ref
{alg:22} is irrelevant since adding an edge does not destroy other {\textsf
V}s in the graph.

\begin{figure}[b]

\includegraphics{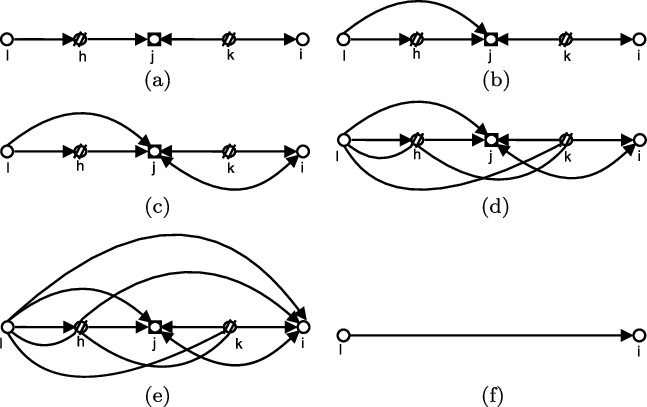}

\caption{(a) A directed acyclic graph $G$, $\margn\in M$
and $\condnc\in C$. (b) The generated graph after applying step 1 of
the table.
(c) The generated graph after applying step 4. (d) The generated graph
after applying step 10. (e) The generated graph after applying step 8.
(f) The generated RG from $G$.}
\label{fig:2ex2}
\end{figure}

Figure~\ref{fig:2ex2} illustrates how to apply Algorithm \ref{alg:22}
step by step to a DAG. We start from step 1 of Table~\ref{tab:21} and proceed step by step. We return to step 1 at the end
if there are any applicable steps left.
Since $\mathcal{DAG}\subset\mathcal{RG}$, one can also use Algorithm
\ref{alg:22} to generate an RG from a DAG. Notice that it is not
enough to simply apply steps 1, 4, and 10 of Table~\ref{tab:21} to a DAG.

\noindent\emph{Global interpretation of the algorithm.} The following
lemma explains the global characteristics of the process of
marginalisation and conditioning.
%
\begin{lemma}\label{lem:22}
Let $H$ be a ribbonless graph. There exists an edge between $i$ and $j$
in the ribbonless graph $\alpha_{\mathrm{RG}}(H;M,C)$ if and only if there
exists an endpoint-identical $m$-connecting path given $M$ and $C$
between $i$ and $j$ in $H$.
\end{lemma}

\noindent\emph{Basic properties of $\alpha_{\mathrm{RG}}$.} We show here that
$\alpha_{\mathrm{RG}}$ is an RG-generating function.
%
\begin{prop}\label{prop:2000}
Graphs generated by Algorithm \ref{alg:22} are RGs.
\end{prop}
%
%
Notice that for every ribbonless graph $H$, it holds that $\alpha_{\mathrm{RG}}(H;\varnothing,\varnothing)=H$.

\noindent\emph{Surjectivity of $\alpha_{\mathrm{RG}}$.} The following result shows
that the class of RGs is the exact class of graph that is generated
after marginalisation and conditioning for DAGs.
%
\begin{prop}\label{prop:vn}
The map $\alpha_{\mathrm{RG}}\dvtx \mathcal{DAG}\rightarrow\mathcal{RG}$ is surjective.
\end{prop}
\subsection{Two necessary properties of RG-generating functions}\label{sec5.2}
Here we establish the two important properties that $\alpha_{\mathrm{RG}}$ (or
every generating function) must have. In short, it must be
well-defined and it must generate a stable class of graphs.

\noindent\emph{Well-definition of $\alpha_{\mathrm{RG}}$.} The following theorem
shows that $\alpha_{\mathrm{RG}}$ is well-defined. This means that instead of
directly generating an RG we can split the nodes that
we marginalise over and condition on into two parts, first generate the
RG related to the first
part, then from the generated RG generate the desired RG related the
second part.
%
\begin{theorem}\label{thm:21n}
For a ribbonless graph $H=(N,F)$ and disjoint subsets $C$, $C_1$, $M$,
and $M_1$ of $N$,
\begin{displaymath}
\alpha_{\mathrm{RG}} \bigl(\alpha_{\mathrm{RG}}(H;M,C);M_1,C_1
\bigr)=\alpha_{\mathrm{RG}}(H;M\cup M_1,C\cup C_1).
\end{displaymath}
\end{theorem}

\noindent\emph{Stability of the graphs generated by $\alpha_{\mathrm{RG}}$.} Here
we introduce the second important property that $\alpha_{\mathrm{RG}}$ must have.
This property is the core idea in defining RGs and in general stable
mixed graphs. The modification applied by the function should generate
a graph that induces the marginal and conditional independence model.
%
\begin{figure}

\includegraphics{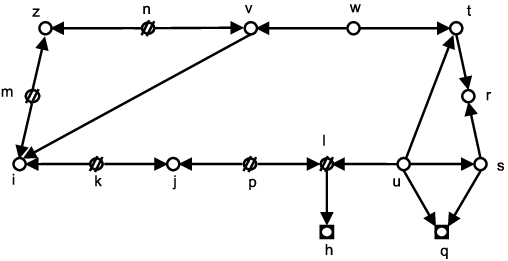}

\caption{A directed acyclic graph $G$ with sixteen nodes,
$\margn\in M$ and $\condnc\in C$.}
\label{fig:parent21}
\end{figure}
%
%
%
\begin{figure}[b]

\includegraphics{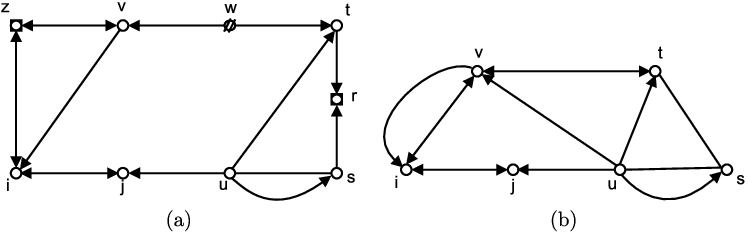}

\caption{(a) The generated ribbonless graph $H=\alpha_{\mathrm{RG}}(G,M,C)$ from the directed acyclic graph $G$ in Figure \protect\ref
{fig:parent21}, $\margn\in M_1$ and $\condnc\in
C_1$. (b) The generated ribbonless graph $\alpha_{\mathrm{RG}}(H,M_1,C_1)$ from $H$.}
\label{fig:example21}
\end{figure}

%
\begin{theorem}\label{thm:21}
For a ribbonless graph $H=(N,F)$ and disjoint subsets $A$, $B$, $C$,
$C_1$, and $M$ of $N$,
\begin{displaymath}
A\dse_mB\cd C_1 \qquad \mbox{in }\alpha_{\mathrm{RG}}(H;M,C)
 \quad \iff \quad  A\dse_mB\cd C\cup C_1 \qquad \mbox{in }H.
\end{displaymath}
\end{theorem}
%
\begin{coro}\label{cor:22}
For a ribbonless graph $H=(N,F)$ and $M$ and $C$ disjoint subsets of $N$,
\begin{displaymath}
\alpha \bigl(\mathcal{J}_m(H);M,C \bigr)=\mathcal{J}_m
\bigl(\alpha_{\mathrm{RG}}(H;M,C) \bigr).
\end{displaymath}
\end{coro}
%
\begin{coro}
The class of RGs is stable.
\end{coro}
The following result has been implicitly discussed in the literature,
for example, see~\citep{cox96}.

\begin{coro}
The classes of UGs and BGs are stable.
\end{coro}
\begin{pf}
The result follows from the fact that, from UGs and BGs, Algorithm \ref
{alg:22} generates UGs and BGs, respectively.
\end{pf}

\begin{example*} Figure~\ref{fig:parent21} illustrates a DAG as
well as two
subsets $M$ and $C$ of its node set. Figure~\ref{fig:example21}(a)
illustrates the
generated ribbonless graph $H$ using Algorithm \ref{alg:22} as well as two
subsets $M_1$ and $C_1$ of its node set,
and Figure~\ref{fig:example21}(b) illustrates the RG generated by the
algorithm from $H$.

For example, consider the graph $\alpha_{\mathrm{RG}}(H,M_1,C_1)$ in Figure~\ref
{fig:example21}(b), and let $A=\{j\}$, $B=\{s\}$ and $C=\{i,u\}$. It is
seen that $v\in\an(C)$.
We have that $A$ is not $m$-separated from $B$ given $C$ since $\langle
j,i,v,t,s\rangle$ is an $m$-connecting path between $A$ and $B$ given
$C$. By
Theorem~\ref{thm:21} we conclude that $A$ is not $m$-separated from
$B$ given $C\cup C_1$. The same conclusion is made by observing
$m$-connecting path $\langle j,i,v,w,t,r,s\rangle$ in $H$.
\end{example*}

\section{Summary graphs}\label{sec6}

\noindent\emph{Definition of summary graphs.} A \emph{summary graph} is a
loopless mixed graph $H=(N,F)$ which contains no $\circ\ful\circ\fla
\circ$ or $\circ\ful\circ\arc\circ$ (arrowhead pointing to line)
and no
direction-preserving cycle as subgraph. Notice that there are also no
multiple edges
in SGs except multiple edges consisting of an arrow
and an arc.

Obviously the class of SGs is a subclass of RGs. Figure~\ref{fig:2exs}
illustrates an SG and an RG that is not an SG. (Because of two reasons:
existence of arrowheads pointing to lines and existence of a double
edge containing line and arrow.)

\begin{figure}

\includegraphics{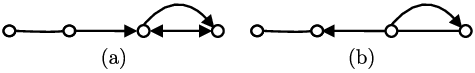}

\caption{(a) An SG. (b) An RG that is not an SG.}
\label{fig:2exs}
\end{figure}
%
\subsection{Generating summary graphs}\label{sec6.1}

\noindent\emph{A local algorithm to generate SGs.} We now present a local
algorithm (after determining the ancestor set of the conditioning set)
to generate an SG from an SG.
%
\begin{alg}\label{alg:23}
$\alpha_{\mathrm{SG}}(H;M,C)$: (Generating an SG from a summary graph
$H$)

Start from $H$. Label the nodes in $\an(C)$.
\begin{enumerate}
\item Apply Algorithm \ref{alg:22}.
\item Remove all edges (arrows or arcs) with arrowhead pointing to a
node in $\an(C)$, and replace these by the edge with the arrowhead
removed (line or arrow) if the edge does not already
exist.
\end{enumerate}

Continually apply step 1 until it is not possible to apply the given
step further before moving to the second step.
\end{alg}
Figure~\ref{fig:2ex3} illustrates how to apply Algorithm \ref
{alg:23} step by step to a DAG. Notice that as it is stated in
the description of the algorithm the order of applying the steps does
matter here.

\begin{figure}

\includegraphics{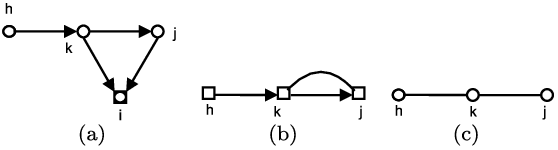}

\caption{(a) A directed acyclic graph $G$, $\condnc\in C$.
(b) The generated graph after applying step 1, $\squarea\in\an(C)$.
(c) The generated SG after applying step 2.}
\label{fig:2ex3}
\end{figure}
%

\noindent\emph{The map $\alpha_{\mathrm{SG}}$ and its basic properties.} For SGs,
we prove analogous results to those for RGs.
%
\begin{prop}\label{pro:32}
Graphs generated by Algorithm \ref{alg:23} are SGs.
\end{prop}

\noindent\emph{The map $\alpha_{\mathrm{RG}.\mathrm{SG}}$ and its properties.} Notice that
step 1 of Algorithm \ref{alg:23} generates an RG before removing the
nodes in $C$. Hence, step 2 of the algorithm generates an SG from an RG
and some extra nodes that are conditioned on. We denote these two steps
by $\alpha_{\mathrm{RG}.\mathrm{SG}}$. This shows that for generating RGs from SGs, $\an
(C)$ is needed.
%
\begin{prop}\label{prop:35nn}
Let $H=(N,E)$ be a ribbonless graph and $M$ and $C$ be subsets of $N$.
It holds that $\alpha_{\mathrm{SG}}(H;M,C)=\alpha_{\mathrm{RG}.\mathrm{SG}}(\alpha_{\mathrm{RG}}(H;M,C);\an(C))$.
\end{prop}

\noindent\emph{Surjectivity of $\alpha_{\mathrm{SG}}$.} The following result shows
that every member of $\mathcal{SG}$ can be generated by a DAG after
marginalisation and conditioning.
%
\begin{prop}\label{prop:vvn}
The map $\alpha_{\mathrm{SG}}\dvtx \mathcal{DAG}\rightarrow\mathcal{SG}$ is surjective.
\end{prop}
%
\subsection{Two necessary properties of SG-generating functions}\label{sec6.2}
Here, we express two important results that have been introduced for
graphs generated by
$\alpha_{\mathrm{RG}}$ for graphs generated by $\alpha_{\mathrm{SG}}$.

\noindent\emph{Well-definition of $\alpha_{\mathrm{SG}}$.} This property is
analogous to the well-definition of $\alpha_{\mathrm{RG}}$ as defined in the
previous section. For a proof based on matrix representations of graphs
and on
properties of corresponding matrix operators, see \citep{wer08}.
%
\begin{theorem}\label{prop:25}
For a summary graph $H=(N,F)$ and disjoint subsets $C$, $C_1$, $M$, and
$M_1$ of $N$,
\begin{displaymath}
\alpha_{\mathrm{SG}} \bigl(\alpha_{\mathrm{SG}}(H;M,C);M_1,C_1
\bigr)=\alpha_{\mathrm{SG}}(H;M\cup M_1,C\cup C_1).
\end{displaymath}
\end{theorem}
%
%
%
%
%
%
%
%

\begin{figure}

\includegraphics{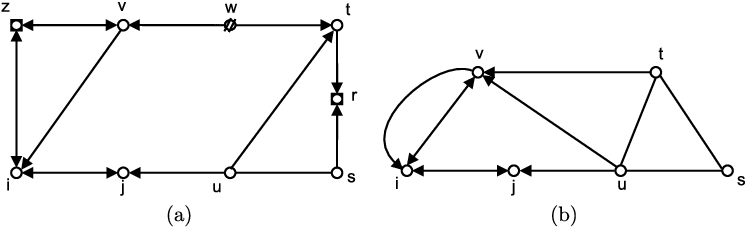}

\caption{(a) The generated SG from the DAG in Figure \protect\ref
{fig:parent21}, $\margn\in M_1$ and $\condnc\in C_1$.
(b) The generated SG from the SG in
(a).}
\label{fig:example22}
\end{figure}

\noindent\emph{Stability of the graphs generated by $\alpha_{\mathrm{SG}}$.} We
prove that analogous to RGs, graphs generated by $\alpha_{\mathrm{SG}}$ induce
the marginal and conditional independence model. This result can be
implied from what was discussed in \citep{wer08}.
%
\begin{theorem}\label{thm:22}
For a summary graph $H=(N,F)$ and disjoint subsets $A$, $B$, $C$,
$C_1$, and $M$ of $N$,
\begin{displaymath}
A\dse_mB\cd C_1 \qquad \mbox{in }\alpha_{\mathrm{SG}}(H;M,C)
 \quad \iff  \quad A\dse_mB\cd C\cup C_1 \qquad \mbox{in }H.
\end{displaymath}
\end{theorem}
\begin{coro}\label{cor:23}
For a summary graph $H=(N,F)$ and $M$ and $C$ disjoint subsets of $N$,
\begin{displaymath}
\alpha \bigl(\mathcal{J}_m(H);M,C \bigr)=\mathcal{J}_m
\bigl(\alpha_{\mathrm{SG}}(H;M,C) \bigr).
\end{displaymath}
\end{coro}
%
\begin{coro}
The class of SGs is stable.
\end{coro}

\begin{example*} Figure~\ref{fig:example22}(a) illustrates the
generated SG from the DAG in Figure~\ref{fig:parent21}
using Algorithm \ref{alg:23} as well as the two subsets $M_1$ and
$C_1$ of its node set. Figure~\ref{fig:example22}(b) illustrates the SG
generated by the algorithm from the SG in part (a).
\end{example*}

\section{Ancestral graphs}\label{sec7} An \emph{ancestral graph} (AG) is a
simple mixed graph that has the following properties for every node
$i$:
\begin{enumerate}
\item$i \notin\an(\pa(i)\cup\spo(i))$;
\item If $\nei(i)\neq\varnothing$, then $\pa(i)\cup\spo
(i)=\varnothing$.
\end{enumerate}
This means that there is no arrowhead pointing to a line and there is
no direction-preserving cycle, and there is no arc with one endpoint
that is an ancestor of the other endpoint in the graph.

AGs are obviously a subclass of SGs, and therefore RGs. Figure~\ref
{fig:2exa1} illustrates an SG that is not ancestral. (Because of an arc
with one endpoint that is an ancestor of the other.)

\begin{figure}

\includegraphics{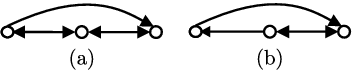}

\caption{(a) An AG. (b) An SG that is not ancestral.}
\label{fig:2exa1}
\end{figure}
%
\subsection{Generating ancestral graphs}\label{sec7.1}

\noindent\emph{A local algorithm to generate AGs.} In \citep{ric02},
there is a method to generate AGs (in fact maximal AGs) globally by
looking at the so-called inducing paths. Here we introduce an algorithm
to generate AGs locally (after determining the ancestor set) by looking
only for {\textsf V}s after determining the ancestor set of the
conditioning set.
%
\begin{alg}\label{alg:25}
$\alpha_{\mathrm{AG}}(H;M,C)$ (Generating an AG from an ancestral graph
$H$):

Start from $H$.
\begin{enumerate}
\item Apply Algorithm \ref{alg:23}.
\item Generate respectively an arrow from $j$ to $i$ or an arc between
$i$ and $j$ for {\textsf V} $j\fra k\arc i$ or {\textsf V} $j\arc k\arc i$ when
$k\in\an(i)$ if the arrow or the arc does not already exist.
\item Remove the arc between $j$ and $i$ in the case that $j\in\an
(i)$, and replace it by an arrow from $j$
to $i$ if the arrow does not already exist.
\end{enumerate}

Continually apply each step until it is not possible to apply the given
step further before moving to the next step.
\end{alg}
Figure~\ref{fig:2ex4n} illustrates how to apply Algorithm \ref
{alg:25} step by step to a DAG.

\begin{figure}[b]

\includegraphics{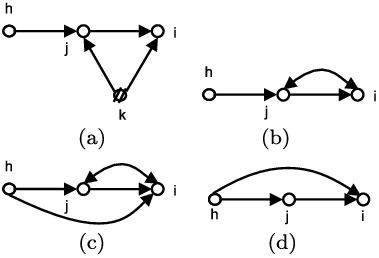}

\caption{(a) A directed acyclic graph $G$, $\margn\in M$.
(b) The generated graph after applying step 1.
(c) The generated graph after applying step 2 for {\textsf V} $\langle
h,j,i\rangle$. (d) The generated AG from $G$ after applying step 3.}
\label{fig:2ex4n}
\end{figure}

\noindent\emph{The map $\alpha_{\mathrm{AG}}$ and its basic properties.} Basic
properties of Algorithm \ref{alg:25} and its corresponding function
are analogous to the basic properties of RGs and SGs.
%
\begin{prop}\label{prop:26}
Graphs generated by Algorithm \ref{alg:25} are AGs.
\end{prop}
%
As before we consider
$\alpha_{\mathrm{AG}}$ as a function from the set of AGs and two subsets
of their node set to the set of AGs.

Notice that by the extension of the generated AG to a maximal AG (as
explained in \citep{ric02}) the same maximal AG as that generated by
the method explained in \citep{ric02} is generated, and hence these
two graphs induce the same independence model. This also explains the
global interpretation of the algorithm. We will not give the details in
this paper.

\noindent\emph{The map $\alpha_{\mathrm{SG}.\mathrm{AG}}$ and its properties.} Notice that
step 1 of Algorithm \ref{alg:25} generates an SG. Hence steps 2 and 3
of the algorithm generate an AG from an SG. We denote these two steps
by $\alpha_{\mathrm{SG}.\mathrm{AG}}$, a function from $\mathcal{SG}$ to $\mathcal{AG}$.
%
\begin{prop}\label{prop:35}
It holds that $\alpha_{\mathrm{AG}}=\alpha_{\mathrm{SG}.\mathrm{AG}}\circ\alpha_{\mathrm{SG}}$.
\end{prop}

\noindent\emph{Surjectivity of $\alpha_{\mathrm{AG}}$.} The following result shows
that every member of $\mathcal{AG}$ can be generated by a DAG after
marginalisation and conditioning.
%
\begin{prop}
The map $\alpha_{\mathrm{AG}}\dvtx \mathcal{DAG}\rightarrow\mathcal{AG}$ is surjective.
\end{prop}
\begin{pf}
The result follows from Proposition~\ref{prop:vvn}, the fact that
$\mathcal{AG}\subseteq\mathcal{SG}$, and if $H\in\mathcal{AG}$
then $\alpha_{\mathrm{SG}.\mathrm{AG}}(H)=H$.
\end{pf}
%
\subsection{Two necessary properties of AG-generating functions}\label{sec7.2}
Again we discuss the two important
properties that we have proven for two other stable mixed graphs.

\noindent\emph{Well-definition of $\alpha_{\mathrm{AG}}$.} Well-definition of
$\alpha_{\mathrm{AG}}$ is analogous to the well-definition of $\alpha_{\mathrm{RG}}$
and $\alpha_{\mathrm{SG}}$
as defined in the previous sections.
%
\begin{theorem}\label{thm:25}
For an ancestral graph $H=(N,F)$ and disjoint subsets $C$, $C_1$, $M$,
and $M_1$ of $N$,
\begin{displaymath}
\alpha_{\mathrm{AG}} \bigl(\alpha_{\mathrm{AG}}(H;M,C);M_1,C_1
\bigr)=\alpha_{\mathrm{AG}}(H;M\cup M_1,C\cup C_1).
\end{displaymath}
\end{theorem}

\noindent\emph{Stability of the graphs generated by $\alpha_{\mathrm{AG}}$.}
Analogous to RGs and SGs,
graphs generated by $\alpha_{\mathrm{AG}}$ induce marginal and conditional
independence models. An analogous result was proven in \citep{ric02}
for maximal AGs that were generated in that paper.
%
\begin{theorem}\label{thm:23}
For an ancestral graph $H=(N,F)$ and disjoint subsets $A$, $B$, $C$,
$C_1$, and $M$ of $N$,
\begin{displaymath}
A\dse_mB\cd C_1 \qquad \mbox{in }\alpha_{\mathrm{AG}}(H;M,C)
 \quad \iff \quad  A\dse_mB\cd C\cup C_1 \qquad \mbox{in }H.
\end{displaymath}
\end{theorem}
%
\begin{coro}[(\citep{ric02})]\label{cor:24}
For an ancestral graph $H=(N,F)$ and $M$ and $C$ disjoint subsets of $N$,
\begin{displaymath}
\alpha \bigl(\mathcal{J}_m(H);M,C \bigr)=\mathcal{J}_m
\bigl(\alpha_{\mathrm{AG}}(H;M,C) \bigr).
\end{displaymath}
\end{coro}
%
\begin{coro}
The class of AGs is stable.
\end{coro}

\begin{example*} Figure~\ref{fig:example23}(a) illustrates the
AG generated from the DAG in Figure~\ref{fig:parent21}. Figure~\ref{fig:example23}(b) illustrates the AG generated by the algorithm
from the AG in part (a).
\end{example*}

\begin{figure}

\includegraphics{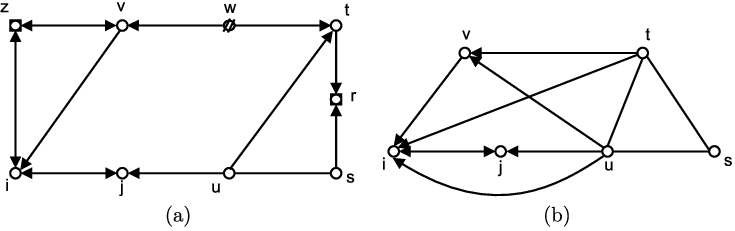}

\caption{(a) The generated AG from the DAG in Figure \protect\ref
{fig:parent21}, $\margn\in M_1$ and $\condnc\in C_1$.
(b) The generated AG from the AG in
(a).}
\label{fig:example23}
\end{figure}
\section{The relationship between different types of stable mixed graphs}\label{sec8}
Thus far, we have defined RGs (as a modification of MCGs), SGs, and
AGs, and introduced algorithms to generate each of these from a graph
of the same class or a DAG, and some algorithms that act between these
classes. Despite the similarities of these definitions and generating
algorithms of these different classes, as well as the parallel theory
developed for these, it is of interest to investigate the exact
relationship between these types of graphs.

\noindent\emph{Corresponding stable mixed graphs.} When one starts from a
DAG and generates different types of stable mixed graphs after
marginalisation over and conditioning on two specific subsets of the
node set of the DAG, the generated graphs must induce the same
independence models.
This leads us to the definition of corresponding stable mixed graphs.
For a directed acyclic graph $G$ and two disjoint subsets of its node
set $M$ and
$C$, graphs $\alpha_{\mathrm{RG}}(G;M,C)$, $\alpha_{\mathrm{SG}}(G;M,C)$, and
$\alpha_{\mathrm{AG}}(G;M,C)$ are called, respectively, the
\emph{corresponding} RG, SG, and AG.

We observe that the corresponding RGs, SGs, and AGs of a DAG induce the
same independence model. This fact, without being formulated in this
way, was discussed in all three papers that define these graphs \citep
{kos02,ric02,wer08}.
%
\begin{prop}
For a directed acyclic graph $G=(V,E)$ and disjoint subsets $C$ and $M$
of $V$,
\begin{displaymath}
\mathcal{J}_m \bigl(\alpha_{\mathrm{RG}}(G;M,C) \bigr)=
\mathcal{J}_m \bigl(\alpha_{\mathrm{SG}}(G;M,C) \bigr)=
\mathcal{J}_m \bigl(\alpha_{\mathrm{AG}}(G;M,C) \bigr).
\end{displaymath}
\end{prop}
\begin{pf}
The result follows from Corollaries \ref{cor:22}, \ref{cor:23}
and \ref{cor:24}.
\end{pf}
As it was shown, in SGs and AGs there are extra properties regarding
the structure of the graph. We know that $\mathcal{AG}\subset\mathcal
{SG}\subset\mathcal{RG}$. The corresponding AG to an SG can be
generated by $\alpha_{\mathrm{SG}.\mathrm{AG}}$ as outlined in Proposition~\ref
{prop:35}. However, we cannot generate the corresponding SG to an RG by
only knowing the RG and not the DAG (or the conditioning set of the
DAG). For example, DAGs $\circ\fla\margn\fra\circ\fra\condnc\fla
\margn$ and $\circ\fla\margn\fra\circ$, where $\margn\in M$ and
$\condnc\in C$, give the same RG $\circ\arc\circ$ but different SGs
$\circ\fla\circ$ and $\circ\arc\circ$ respectively. This is also
true for AGs instead of SGs.

It is possible, however, to introduce an algorithm to generate SGs that
induce the same independence model as the given RGs, by removing
arrowheads pointing to a line or a node that is an ancestor of a node
that is the endpoint of a line.

\begin{figure}

\includegraphics{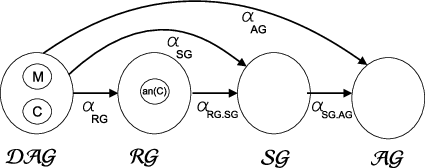}

\caption{The relationship between $\mathcal{DAG}$,
$\mathcal{RG}$, $\mathcal{SG}$ and $\mathcal{AG}$.}
\label{fig:32}
\end{figure}

We have also seen that the image of generating functions is big enough
to cover all graphs included in the set of the related type of stable
mixed graphs, since the generating functions are surjective. On the
other hand, it is easy to show that generating functions are not
injective. Therefore, the relationship between the three types of
stable mixed graphs is summarised by the diagram in Figure~\ref
{fig:32}, in which one can only move towards the directions that arrows show.

%
\section{Discussion on the use of different types of stable mixed graphs}\label{sec9}
By what we discussed, if $G$ is a DAG with latent variables $M$ and
selection variables $C$ then stable mixed graphs are a class of graphs
that represent the independence model implied among the remaining
variables, conditional on the selection variables. However, each of the
three types has been used in different contexts and for different purposes.

\noindent\emph{Why MCGs or RGs?} MCGs have been introduced in order to
straightforwardly deal with the problem of finding a class of graphs
that is closed under marginalisation and conditioning by a simple
process of deriving these from DAGs. In fact, the class of MCGs is much
larger than what one really needs for representing independence models
after marginalisation and conditioning. We have noted that only MCGs
that are ribbonless can be generated this way.

\noindent\emph{Why SGs?} The main goal of defining SGs is to trace the
effects after marginalisation and conditioning, as will be explained
shortly in this section. By using binary matrix
representations of graphs, called edge matrices, and corresponding
matrix operators \citep{wer06}, the edge matrix of a SG is obtained.
It contains
three types of edge matrices: those for solid lines, dashed lines
(corresponding to arcs),
and for arrows. In the family of joint Gaussian distributions, solid lines
in concentration graphs correspond to concentration matrices, dashed
lines in covariance graphs to covariance matrices and arrows
to equation parameters in structural equation models.

SGs are used when the generating DAG is known. Despite knowledge on the
structure of the generating DAG, SGs are still of interest in at least
three situations: (1) For models with large number of unobserved and
selection variables; and (2) for the comparison of models when one of
them has unobserved or selection variables that are a subset of the
unobserved or selection variables of the other; (3) for detecting some
types of confounding as shown in \citep{werc08} and as described
briefly later.

\noindent\emph{Why AGs?} The main goal of defining AGs is to represent and
parametrise sets of distributions obeying
Markov properties. Even though, we discussed the class of AGs in this
paper to sustain a parallel theory to RGs and SGs, the class of maximal
AGs possess some desired
properties that AGs do not. These include the fact that under the
Gaussian path diagram parametrisation the maximal AG only implies
independence constraints, while a general AG implies other types of
constraints. We will give a short discussion on maximality in this
section. Maximal AGs are the simplest structures that capture the
modified independence model, and are also of interest when the
generating DAG is not known, but a set of conditional independencies is
known. In the Gaussian case, maximal AGs are
identified. In contrast to DAG models with hidden variable, the models
are curved exponential families \citep{ric02}, and conditional fitting
algorithm for maximum likelihood
estimation exists \citep{drt04}.

\noindent\emph{Maximal stable mixed graph.} A graph $G$ is called \emph
{maximal} if by adding any edge to $G$ the independence model induced
by $\dse_m$ changes (gets smaller). Therefore, in maximal graphs,
every missing edge corresponds to at least one independence statement
in the induced independence model. This leads to validity of a
so-called pairwise Markov property.

In \citep{ric02}, maximality of the subclass of AGs was studied. This
result also holds for RGs and says that a ribbonless graph $H$ is
maximal if and only if $H$ does not contain any \emph{primitive
inducing paths}, which are paths of form $\langle j,q_1,q_2,\ldots,q_p,i\rangle$, on which $i\nsim j$ and for every $n$, $1\leq n \leq
p$, $q_n$ is a collider on the path and $q_n\in\an(\{i\}\cup\{j\})$.
We shall not give the details in this paper.

Therefore, to generate a maximal stable mixed graph from a stable mixed
graph one should repeatedly generate arrows from $j$ to $i$ for
primitive inducing paths between non-adjacent $i$ and $j$ where there
is no arrowhead pointing to $j$, and generate arcs between $i$ and $j$
for primitive inducing paths between non-adjacent $i$ and $j$ where
there are arrowheads pointing to $i$ and $j$. Notice that by applying
this algorithm after the generating algorithms one can generate a
maximal AG, SG, or RG.

As discussed, maximal AGs possess many desired properties that AGs do
not. For SGs, it is conjectured that maximal SGs possess the same
statistical properties that both maximal AGs and SGs do possess. To
show this, further work is needed.

\noindent\emph{The structure of different types of stable mixed graphs.}
If we suppose that stable mixed graphs are only used to represent the
independence model after marginalisation and conditioning, then we can
consider all types as equally appropriate. The question then will be
reduced to how simple or fast generating a type of graph is. We have
seen that AGs have the simplest structure among the three types of
stable mixed graphs, and RGs are the most complex. Therefore, as we
have also seen, it is more complex to generate an AG than to generate
an SG, and to generate an SG than to generate an RG. On the other hand,
the simpler structure allows a faster way of checking independence
statements. Hence, it is a tradeoff that depends on the relative size
of the marginalisation and conditioning sets in graphs.

When generating stable mixed graphs from DAGs, one always loses some
information in order to obtain a simpler structure in stable mixed
graphs. RGs have lost the least information among the three types of
stable mixed graphs, while AGs the most. Here we discuss the lost
information in the context of regression analysis.

\noindent\emph{Multivariate regression and stable mixed graphs.} The
problem of constructing stable mixed graphs
was originally posed by \citep{wer94} in the context of multivariate
statistics based on regression analysis. In such literature, the DAG
model is defined by sequences
of univariate recursive regressions, called a \emph{linear triangular
system} by \citep{werc04}, that is, for $i=1,\ldots,d_N-1$, each
single response variable $Y_i$
is regressed on $Y_{\pa(i)}$, where the parents of $i$ are a subset of
$\{i+1,\ldots,d_N\}$. Linear triangular systems can be written as
$AY=\epsilon$, where $A$ is an upper-triangular matrix with unit
diagonal elements, and $\epsilon$ is a vector of zero mean and
uncorrelated random variables, called \emph{residuals}. Here the
nonzero regression coefficient of $Y_i$ on $Y_j$ can be attached the
arrow from $j$ to $i$ in the DAG and is called the \emph{direct
effect} of $Y_j$ on $Y_i$; see \citep{cox93}.

In particular, for linear triangular systems, RGs alert to distortions
due to so-called over-conditioning via multiple edges consisting of a
line and an arrow. Over-conditioning arises by conditioning on a
variable that is a response of two variables, one of which itself is a
response to the other one.

For example, in Figure~\ref{nnnn}, the generating process is given by
three linear equations,
\[
Y_1=\beta Y_2+\delta Y_3+
\epsilon_1, \qquad  Y_2=\gamma Y_3+
\epsilon_2, \qquad  Y_3=\epsilon_3,
\]
where each residual $\epsilon_i$ has mean zero and is uncorrelated
with the explanatory variables on the
right-hand side of the equation.

\begin{figure}

\includegraphics{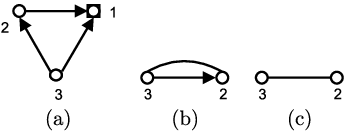}

\caption{(a) A directed acyclic graph $G$ with node 1 to be
conditioned on. (b) The RG generated from $G$.
(c) The SG or AG generated from $G$.}
\label{nnnn}
\end{figure}
By conditioning on $Y_1$, the conditional
dependence of $Y_2$ on only $Y_3$ is obtained, which consists of the
direct effect $\gamma$ and
an indirect effect of $Y_2$ on $Y_3$ via $Y_1$. This may be seen by
direct calculation,
assuming that the residuals $\epsilon_i$ have a Gaussian distribution,
which leads to
\[
E(Y_2\cd Y_3)= \bigl(\gamma- \bigl\{ \bigl(1-
\gamma^2 \bigr)/ \bigl(1-\rho_{13}^2 \bigr)
\bigr\} \beta\rho_{13} \bigr)Y_2, \qquad  \mbox{where }
\rho_{13}= \delta+\beta\gamma.
\]
Thus, the direct effect $\gamma$ is distorted by $-\{(1-\gamma^2)/(1-\rho_{13}^2)\}\beta\rho_{13}$. The potential presence of this
distortion is represented in (b) by the addition of an arrow.

In addition, the existence of multiple edges with an arrow and an arc,
and arcs with one endpoint ancestor of the other, which are not
permissible in AGs, respectively, alerts distortions due to so-called
direct and indirect \emph{confounding}.

\begin{figure}

\includegraphics{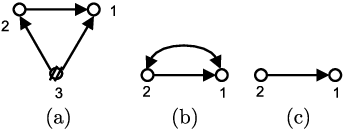}

\caption{(a) A directed acyclic graph $G$ with node 3 to be
marginalised over. (b) The SG generated from $G$.
(c) The AG generated from $G$.}
\label{fig:321}
\end{figure}

With the same generating process as explained for Figure~\ref{nnnn}, by
integrating out $Y_3$ in Figure~\ref{fig:321}, the conditional
dependence of $Y_1$ on only $Y_2$ is obtained, which consists of the
direct effect $\beta$ and
an indirect effect of $Y_1$ on $Y_2$ via $Y_3$. This leads to
\[
E(Y_1\cd Y_2)=(\beta+\delta\gamma)Y_2.
\]
Thus, the direct effect $\beta$ is distorted by $\delta\gamma$. The
potential presence of this distortion is represented in (b) by the
addition of an arc. This example indicates a distortion due to direct
confounding; see \citep{werc08}. Indirect confounding was also studied
in \citep{werc08} for marginalising only over a full set
of background variables and also in \citep{wer08} more generally
relating SGs
to corresponding maximal AGs.

\begin{appendix}
\section*{Appendix: Proofs}\label{appm}
Here we present the proof of lemmas, propositions, and theorems of this
paper, but first we introduce some observations that are used in our
proofs as the following lemmas.
%
\begin{lemma}\label{lem:22nn}
If $i\in\an(j)$ in $\alpha_{\mathrm{RG}}(H;M,C)$, then in $H$ one of the
following holds: \emph{(1)}~$i\in\an(j)$; \emph{(2)} $i$ or a descendant of $i$ is
the endpoint of a line; \emph{(3)} $i\in\an(C)$.
\end{lemma}
\begin{pf}
We know that there is a direction-preserving path $\pi=\langle
i=i_0,i_1,\ldots,i_p=j\rangle$ in $\alpha_{\mathrm{RG}}(H;M,C)$. Consider the
$i_0i_1$-edge. By Lemma~\ref{lem:22} in $H$ given $M$ and $C$, there
is an $m$-connecting path between $i_0$ and $i_1$, on which there is no
arrowhead pointing to $i_0$. One can observe that if this path is not a
direction-preserving path then one of the following holds: (1)~$i_0$ is
an ancestor of a collider node on the path, which is in $C\cup\an
(C)$. Hence, $i\in\an(C)$; (2) $i_0$ is the endpoint of a line or an
ancestor of a node that is the endpoint of a line on the path. If (1)
or (2) holds, then we are done, hence assume that $i_0\in\an(i_1)$.
By the same argument and by induction along the nodes of $\pi$, we
conclude the result.
\end{pf}
%
\begin{lemma}\label{lem:22nnn}
For $i$ and $j$ outside $M\cup C$, if $i\in\an(j)$ in $H$ then one of
the following holds: \emph{(1)}~$i\in\an(j)$ in $\alpha_{\mathrm{RG}}(H;M,C)$; \emph{(2)}
$i\in\an(C)$ in $H$.
\end{lemma}
\begin{pf}
We know that there is a direction-preserving path $\pi=\langle
i=i_0,i_1,\ldots,i_p=j\rangle$ in $H$. Consider the $i_0i_1$-edge. We
now have three cases: (1) If $i_1\in C$, then $i\in\an(C)$ in $H$ and
we are done. (2) If $i_1\in M$, then Algorithm \ref{alg:22} generates
an arrow from $i_0$ to $i_2$. (3) If $i_1\not\in M\cup C$, then $i_0\in
\an(i_2)$ in $\alpha_{\mathrm{RG}}(H;M,C)$. By the same argument and by
induction along the nodes of $\pi$, we conclude the result.
\end{pf}
The following lemma deals with the \emph{concatenation} of
$m$-connecting paths. We shall not give the details of the proof here;
see \citep{sad12}.
%
\begin{lemma}\label{lem:2j2}
In an RG, suppose that given $M$ and $C$ there are $m$-connecting paths
$\langle i=i_0,i_1,\ldots,i_n,h\rangle$ between $i$ and $h$ and
$\langle j=j_0,j_1,\ldots,j_m,h\rangle$
between $h$ and $j$. In this case, there is an $m$-connecting path
given $M$ and $C$ between $i$ and $j$ if one of the following holds:
\begin{enumerate}[(b2)]
\item[(a1)] $\langle i_n,h,j_m\rangle$ is collider and $h\in C\cup\an(C)$;
\item[(a2)] $i_n=j_m$ with arrowhead pointing to $h$ on the $i_nh$-edge
and $h\in C\cup\an(C)$;
\item[(b1)] $\langle i_n,h,j_m\rangle$ is non-collider and $h\in M$;
\item[(b2)] $i_n=j_m$ with no arrowhead pointing to $h$ on the
$i_nh$-edge and $h\in M$.
\item[(c1)] $\langle i_n,h,j_m\rangle$ is collider and $h$ or a
descendant of $h$ is the endpoint of a line or a direction-preserving cycle;
\item[(c2)] $i_n=j_m$ with arrowhead pointing to $h$ on the $i_nh$-edge
and $h$ or a descendant of $h$ is the endpoint of a line or a
direction-preserving cycle.
\end{enumerate}
\end{lemma}
\begin{pf*}{Proof of Proposition~\ref{prop:2000}}
Graphs generated by Algorithm \ref{alg:22} have obviously three
desired types of edges and are loopless.

Now suppose, for contradiction, that there is a ribbon $\langle
i,h,j\rangle$ in a generated graph $\alpha_{\mathrm{RG}}(H;M,C)$. By Lemma~\ref{lem:22} in $H$ given $M$ and $C$, there are $m$-connecting paths
$\pi_1=\langle i=i_0,i_1,\ldots,i_n,h\rangle$ between $i$ and $h$ and
$\pi_2=\langle j=j_0,j_1,\ldots,j_m,h\rangle$ between $h$ and $j$
such that there are arrowheads at $h$ on both $i_nh$- and $j_mh$-edges.
(Notice that it is possible that $i_n=j_m$ and it is also possible that
$i_n=i$ or $j_m=j$ in $H$.)

We also know that, in $\alpha_{\mathrm{RG}}(H;M,C)$, the node $h$ is the
endpoint of a line or on a direction-preserving cycle or there is a
direction-preserving path \mbox{$\pi=\langle h=h_0,h_1,\ldots,h_p=k\rangle$}
from $h$ to $k$ such that $k$ is the endpoint of a line or on a
direction-preserving cycle. Now we consider two cases. In case I, we
suppose that such a $\pi$ does not exist and in case II we suppose
that such a $\pi$ exists.

\emph{Case I.} In case I.1, we suppose that $h$ is the endpoint of a
line and in case I.2 we suppose that $h$ is on a direction-preserving cycle.

\emph{Case I.1.} Suppose that $h$ is the endpoint of an $hl$-line in
$\alpha_{\mathrm{RG}}(H;M,C)$. By Lemma~\ref{lem:22} in $H$ given $M$ and $C,$
there is an $m$-connecting path between $h$ and $l$, on which there is
no arrowhead pointing to $h$ or $l$. One can observe that $h$ is an
ancestor of either (1) a collider node on the path or (2) a node that
is the endpoint of a line on the path. Thus, we have the two following cases:

(1) If $h$ is an ancestor of a collider node $t$, then $h\in\an(C)$
in $H$ since $t\in C\cup\an(C)$. Hence, by Lemma~\ref{lem:2j2}(a),
there is an $m$-connecting path given $M$ and $C$ between $i$ and $j$
in $H$.

(2) If $h$ is an ancestor of a node that is the endpoint of a line on
the path, then by Lemma~\ref{lem:2j2}(c) there is an $m$-connecting
path given $M$ and $C$ between $i$ and $j$ in $H$.

By Lemma~\ref{lem:22}, both cases imply that $i\sim j$ in $\alpha_{\mathrm{RG}}(H;M,C)$ and the $ij$-edge is endpoint-identical to the
$m$-connecting path. Therefore, $\langle i,h,j\rangle$ is not a
ribbon, a contradiction.

\emph{Case I.2.} Suppose that $h$ is on a direction-preserving cycle
in $\alpha_{\mathrm{RG}}(H;M,C)$. By Lemma~\ref{lem:22nn} in $H$ one of the
following holds: (1) $h\in\an(h)$; (2) $h$ or a descendant of $h$ is
the endpoint of a line; (3) $h\in\an(C)$. Cases (2) and (3) lead to
contradiction as explained in case I.1. Therefore, suppose that $h$ is
on a direction-preserving cycle in $H$. This by Lemma~\ref{lem:2j2}(c)
implies that there is an $m$-connecting path given $M$ and $C$ between
$i$ and $j$ in $H$, which implies that $\langle i,h,j\rangle$ is not a
ribbon. This is a contradiction.

\emph{Case II.} By Lemma~\ref{lem:22nn} in $H$ one of the following
holds: (1) $h\in\an(k)$; (2) $h$ or a descendant of $h$ is the
endpoint of a line; (3) $h\in\an(C)$. Cases (2) and (3) lead to
contradiction as explained in case I.1. Hence, it holds that $h\in\an
(k)$ in $H$. This together with the same argument as that of case I
(for $k$ instead of $h$) leads to a contradiction.
\end{pf*}
\begin{pf*}{Proof of Lemma~\ref{lem:22}} ($\Rightarrow$) If an edge
between $i$ and $j$ in $\alpha_{\mathrm{RG}}(H;M,C)$ does
not exist in $H$, then it has been generated by certain intermediate
graphs that have each been generated by adding one edge to the previous
graph by one of
the steps of Table~\ref{tab:21}. We denote these graphs by the sequence
$\langle H=H_0,H_1,\ldots,H_n,\alpha_{\mathrm{RG}}(H;M,C)\rangle$, where $H_n$
is the last step before removing $M$ and $C$.

We prove by reverse induction on $p$ that in all $H_p$, $0\leq p\leq
n$, between $i$ and $j$ there exists a path on which non-collider inner
nodes are in $M$ and collider inner nodes or their descendants are
either in $C$ or the endpoint of a line. For $p=n$, there is obviously
an edge between $i$ and $j$. We show that if there is such a path in
$H_r$ then we can find the same type of path between $i$ and $j$ in $H_{r-1}$.

If all edges along the path exist in $H_{r-1}$, then we should check
that a collider node that is an ancestor of a member of $C$ or an
ancestor of a node that is the endpoint of a line in $H_r$ is an
ancestor of a member of $C$ or an ancestor of a node that is the
endpoint of a line in $H_{r-1}$. If an arrow has been generated along
the direction-preserving path in $H_r$, then it has been generated by
the {\textsf V}s $\langle i',m,j'\rangle$ of the first three steps or the
{\textsf V} $\langle i',s,j'\rangle$ of step 8 of Table~\ref{tab:21}. If it is step
1, then we can replace the $i'j'$-arrow by $\langle i',m,j'\rangle$ to
obtain a direction-preserving path. If it is steps 2 or 3, then node
$j'$ is the endpoint of a line and we are done. If it is step 8 then,
since $s\in C\cup\an(C)$, the inner node of the {\textsf V} is in $\an
(C)$ and we are done.

Thus, suppose that an $i'j'$-edge along the $m$-connecting path is the
edge that has been generated by this step. This has been generated by
one of {\textsf V}s of Table~\ref{tab:21}. Since in all cases the {\textsf V} is
endpoint-identical to the $i'j'$-edge, and since all inner nodes of the
non-collider {\textsf V}s are in $M$ and all inner nodes of the collider
{\textsf V}s are in $C\cup\an(C)$, by placing the {\textsf V} instead of the
$i'j'$-edge on the path, we still get a path whose non-collider inner
nodes are in $M$ and either whose collider inner nodes are in $C\cup
\an(C)$ or whose collider nodes or a descendant of them are the
endpoint of a line, as required.

Therefore, by reverse induction, there exists a path, as described
above, in $H$. However, since $H$ is ribbonless, the path cannot
contain a collider {\textsf V} $\langle i',h,j'\rangle$ such that $h$ or a
descendant of $h$ is the endpoint of a line unless $i'\sim j'$ and the
$i'j'$-edge is endpoint-identical to $\langle i',h,j'\rangle$. In this
case, the $i'j'$-edge can be used instead of $\langle i',h,j'\rangle$
and, by induction, we obtain an $m$-connecting path given $M$ and $C$
between $i$ and $j$.

The fact that in Table~\ref{tab:21} the {\textsf V}s are endpoint-identical to the
generated $i'j'$-edges implies that all discussed paths in each $H_p$
are endpoint-identical.

($\Leftarrow$) Suppose that there is an $m$-connecting path $\pi$ given
$M$ and $C$ between $i$ and $j$ in $H_k$, $0\leq k\leq n-1$. We prove
as long as $r>0$ if there is an $m$-connecting path given $M$ and $C$
between $i$ and $j$ in $H_k$ with $r$
inner nodes then there is an $m$-connecting path given $M$ and $C$
between $i$ and $j$ in $H_{k+1}$ with $r-1$ inner nodes. By induction
we will finally obtain an $m$-connecting path between $i$ and
$j$ without inner nodes, that is, an edge between $j$ and $i$.

Consider an $m$-connecting path given $M$ and $C$ between $i$ and $j$
in $H_k$ with $r>0$
inner nodes. Consider an arbitrary inner node on the path. If this node
is collider, then one of the {\textsf V}s 8, 9, or 10 of Table~\ref{tab:21} is
employed to generate an edge between the endpoints of the {\textsf V} in
$H_{k+1}$. Since the generated edge is endpoint-identical to the ${\textsf
V}$, one can use the generated edge instead of the {\textsf V} to obtain an
$m$-connecting path with $r-1$ inner nodes. If the arbitrary node is
non-collider then one of the other {\textsf V}s of Table~\ref{tab:21} is used.

It is easy to check that the generated edges are endpoint-identical to
the $m$-connecting paths in the final graph. This implies the result.
\end{pf*}
\begin{pf*}{Proof of Proposition~\ref{prop:vn}}
Let $H=(N_1,F_1)\in\mathcal{RG}$. We
generate a directed graph $G=(V,E)$ from $H$ as follows: We leave
arrows that are not on any direction-preserving cycle unchanged. For
direction-preserving cycles, instead of one arbitrary arrow from $j$ to
$i$ on the cycle we place $j\fra\condnc\fla\margn\fra i$, where
$\margn\in M$ and $\condnc\in C$, and leave all other arrows unchanged.
Instead of an arc between $j$ and $i$ we place a {\textsf V}
between $j$ and $i$ with inner source node in $M$. Instead of a
line between $j$ and $i$ we place a {\textsf V} between $j$ and
$i$ with inner collider node in $C$. The graph $G$ is obviously a directed
graph. Furthermore, all newly generated nodes have degree 2 and
the direction of arrows changes on them, hence these cannot be on
any direction-preserving cycle. In addition, if $i$ and $j$ are in
$N_1$ and
$i\sim j$ in $G$ then $i\in\pa(j)$ or $j\in\pa(i)$ in $H$.
Therefore, the existence of a direction-preserving cycle in $G$
implies the existence of the same direction-preserving cycle in $H$.
But by the nature of the construction of $G$ we know that
direction-preserving cycles in $H$ do not make direction-preserving
cycles in $G$, hence $G$ is acyclic.

We should prove that $\alpha_{\mathrm{RG}}(G;M,C)=H$. Let
$\alpha_{\mathrm{RG}}(G;M,C)=(N_2,F_2)$. Obviously $N_1=N_2$. Suppose that
$i\sim j$ ($j\in\pa(i)$, $j\in\spo(i)$, or $j\in\nei(i)$) in
$H$. Therefore, we have the active alternating path or one of the
{\textsf V}s between $i$ and $j$ that by Algorithm
\ref{alg:22} forms exactly the same type of edge in
$\alpha_{\mathrm{RG}}(G;M,C)$.

Conversely, suppose that $i\sim j$
($j\in\pa(i)$, $j\in\spo(i)$, or $j\in\nei(i)$) in
$\alpha_{\mathrm{RG}}(G;M,C)$. By Lemma~\ref{lem:22} we know that there is an
endpoint-identical $m$-connecting path given $M$ and $C$ in $G$.
Consider a shortest endpoint-identical $m$-connecting path $\pi$.
Since in $G$ there is no transition node in $M$, $\pi$ is active
alternating with respect to $M$ and $C\cup\an(C)$.
If $\pi$ has no collider node in $\an(C) \setminus C$, then by the
nature of the construction of $G$ we know that it has two edges
(if both endpoints are children or parents) or three (if it is from $j$
to $i$ on a direction-preserving cycle) and that it has been generated
by an edge (arrow, arc, or line) in $H$.
Suppose, for contradiction, that there is a collider node $i\in\an
(C)\setminus C$ on $\pi$. We have that $i\in N_1$, and by the process
of generating a DAG explained here, the only place that a node in $C$
has been generated is by a line or an arrow on a direction-preserving
cycle in $H$.
Therefore, $i\in\an(k)$ for a node $k$ that is the endpoint of a line
or is on a direction-preserving cycle. Hence, $H$ contains a ribbon, or
the endpoints of the collider {\textsf V} with $i$ as inner node are
adjacent by an endpoint-identical edge. The former contradicts that $H$
is ribbonless, and the latter contradicts that $\pi$ is shortest.
\end{pf*}
\begin{pf*}{Proof of Theorem~\ref{thm:21n}} ($\Rightarrow$) If
there is an $ij$-edge in $\alpha_{\mathrm{RG}}(\alpha_{\mathrm{RG}}(H;M,C);M_1,C_1)$,
then by Lemma~\ref{lem:22}, there is an $m$-connecting path $\pi
=\langle i=i_0,i_1,\ldots,i_{n-1},i_n=j\rangle$ between $i$ and $j$
given $M_1$ and $C_1$ in $\alpha_{\mathrm{RG}}(H;M,C)$ that is
endpoint-identical to the edge.

For the {\textsf V} $\langle i,i_1,i_2\rangle$ on $\pi$, again by Lemma~\ref{lem:22}, given $M$ and $C$ there are $m$-connecting paths $\pi_1$ between $i$ and $i_1$ and $\pi_2$ between $i_1$ and $i_2$ in $H$.
These paths can be considered $m$-connecting given $M\cup M_1$ and
$C\cup C_1$ and are endpoint-identical to the edges. This implies that
if $i_1$ is collider (or non-collider) on $\pi$ then on the
concatenation of $\pi_1$ and $\pi_2$ it remains collider (or
non-collider), or that there is an arrowhead pointing to it (or no
arrowhead pointing to it) from a joint node on $\pi_1$ and $\pi_2$.
We know that if $i_1$ is non-collider then it is in $M_1$ and if it is
collider then it is in $C_1\cup\an(C_1)$ in $\alpha_{\mathrm{RG}}(H;M,C)$. If
$i_1\in\an(C_1)$ in $\alpha_{\mathrm{RG}}(H;M,C)$ then, by Lemma~\ref
{lem:22nn}, one of the following holds in $H$: (1) it is in $\an(C_1)$;
(2) it is in $\an(C)$; (3) it is the endpoint of line or an ancestor of
a node that is the endpoint of a line. Therefore, by Lemma~\ref
{lem:2j2}, there is an $m$-connecting path between $i$ and $i_2$ given
$M\cup M_1$ and $C\cup C_1$ in $H$, which is endpoint-identical to the
{\textsf V}. By induction along $\pi$ there is an $m$-connecting path
between $i$ and $j$ given $M\cup M_1$ and $C\cup C_1$ in $H$, which is
endpoint-identical to the $ij$-edge. Therefore, by Lemma~\ref{lem:22}
there is the same type of $ij$-edge in $\alpha_{\mathrm{RG}}(H;M\cup
M_1,C\cup C_1)$.

($\Leftarrow$) If there is an edge between $i$ and $j$ in $\alpha_{\mathrm{RG}}(H;M\cup M_1,C\cup C_1)$ then, by Lemma~\ref{lem:22}, there is an
$m$-connecting path $\pi$ given $M\cup M_1$ and $C\cup C_1$ in $H$
that is endpoint-identical to the $ij$-edge. All inner nodes of $\pi$
are either in $M\cup C\cup\an(C)$ or $M_1\cup C_1\cup\an(C_1)$.
Therefore, $\pi$ can be partitioned into $m$-connecting subpaths given
$M$ and $C$ and single nodes, where the endpoints of subpaths and
single nodes are in $M_1\cup C_1\cup\an(C_1)$. Therefore, by Lemma~\ref{lem:22}, in $\alpha_{\mathrm{RG}}(H;M,C)$ the endpoints of each of the
discussed subpaths of $\pi$ are connected by an edge that is
endpoint-identical to the subpath.

In addition, for each collider {\textsf V} $\langle l,k,h\rangle$, where
$k\in\an(C_1)$ in $H$ and by Lemma~\ref{lem:22nnn}, one of the
following holds: (1) $k\in\an(C)$ in $H$; (2) $k\in\an(C_1)$ in
$\alpha_{\mathrm{RG}}(H;M,C)$. Case (1) implies that there is an
endpoint-identical $lh$-edge to the {\textsf V} in $\alpha_{\mathrm{RG}}(H;M,C)$,
which can be used instead of $\langle l,k,h\rangle$ to generate an
$m$-connecting path.

Hence in $\alpha_{\mathrm{RG}}(H;M,C),$ there is an $m$-connecting path given
$M_1$ and $C_1$ between $i$ and $j$, which is endpoint-identical to
$\pi$. Therefore, again by Lemma~\ref{lem:22}, there is the same type
of $ij$-edge in $\alpha_{\mathrm{RG}}(\alpha_{\mathrm{RG}}(H;M,C);M_1,C_1)$.
\end{pf*}
\begin{pf*}{Proof of Theorem~\ref{thm:21}} To prove the result, it
is enough to show that, between nodes $i$ and $j$ outside $C\cup
C_1\cup M$, there is an $m$-connecting path given $C\cup C_1$ in $H$ if
and only if there is an $m$-connecting path given $C_1$ in $\alpha_{\mathrm{RG}}(H;M,C)$.

($\Rightarrow$) Suppose that between $i$ and $j$ there is an
$m$-connecting path given $C\cup C_1$ in $H$. This path can be
partitioned into $m$-connecting subpaths given $C$ and single nodes,
where the endpoints of subpaths and single nodes are colliders in
$C_1\cup\an(C_1)$. In addition, for each collider {\textsf V} $\langle
l,k,h\rangle$, where $k\in\an(C_1)$ in $H$ and by Lemma~\ref
{lem:22nnn}, one of the following holds: (1) $k\in\an(C)$ in $H$; (2)~$k\in\an(C_1)$ in $\alpha_{\mathrm{RG}}(H;\varnothing,C)$. Case (1) implies
that there is an endpoint-identical $lh$-edge to the {\textsf V} in $\alpha_{\mathrm{RG}}(H;\varnothing,C)$, which can be used instead of $\langle
l,k,h\rangle$ to generate an $m$-connecting path.

Hence by Lemma~\ref{lem:22}, in $\alpha_{\mathrm{RG}}(H;\varnothing,C)$ there
is an $m$-connecting path given $C_1$ between $i$ and~$j$. The inner
non-collider nodes on this path are either in $M$ or in $N\setminus
(M\cup C\cup C_1)$. On this path there are subpaths with only
non-collider inner nodes in $M$, that is, $m$-connecting subpaths
given $M$ and $\varnothing$. By Theorem~\ref{thm:21n} after
marginalisation over $M$, $\alpha_{\mathrm{RG}}(H;M,C)$ is obtained, in which,
by Lemma~\ref{lem:22}, the endpoints of each of such subpaths are
connected by an edge that is endpoint-identical to the subpath. In
addition, if a collider node is in $\an(C_1)$ in $\alpha_{\mathrm{RG}}(H;\varnothing,C)$, then by Lemma~\ref{lem:22nnn} (since the
conditioning set is empty and case (2) of the lemma does not hold), the
collider node remains in $\an(C_1)$ in $\alpha_{\mathrm{RG}}(H;M,C)$.
Therefore, between $i$ and $j$ there is an $m$-connecting path given
$C_1$ in $\alpha_{\mathrm{RG}}(H;M,C)$.

($\Leftarrow$) Suppose that between $i$ and $j$ there is an
$m$-connecting path $\pi$ given $C_1$ in $\alpha_{\mathrm{RG}}(H;M,C)$. By
Lemma~\ref{lem:22}, for each edge of $\pi$, there is an
endpoint-identical $m$-connecting path given $M$ and $C$ in $H$, which
is obviously an $m$-connecting path given $C$. On the concatenation of
these paths and for the endpoints of the paths we have the two
following cases: (1) When the endpoints are non-collider or there is no
arrowhead at them on the concatenation, they are non-collider on $\pi$
and therefore outside $M\cup C\cup C_1$; (2) When the endpoints are
collider or there is an arrowhead at them, they are collider on $\pi$
and therefore in $C_1\cup\an(C_1)$. Therefore, by Lemma~\ref
{lem:2j2}, there is an $m$-connecting path given $C\cup C_1$ in $H$.
\end{pf*}
\begin{pf*}{Proof of Proposition~\ref{pro:32}} Firstly, the graph
generated by the algorithm has only the three desired types of edges
and no multiple edge of the same type; therefore, it generates a
mixed graph.

Let $H$ be the input summary graph and $H_0=\alpha_{\mathrm{RG}}(H;M,C)$ be the
generated graph after applying step 1 of Algorithm \ref{alg:23}. Here
in part I we prove that there is no arrowhead pointing to lines, and in
part II we prove that there is no direction-preserving cycle in the
generated graph.

\emph{Part I.} If, for contradiction, there is an arrowhead at $i$ (on
an $ik$-edge) and $i$ is the endpoint of an $ij$-line in the generated
graph then we have the following cases: (1) The $ij$-line is an arrow
from $i$ to $j$ in $H_0$; (2) the $ij$-line is an arrow from $j$ to $i$
in $H_0$; (3) the $ij$-line is an arc in $H_0$; (4) the $ik$-edge is an
arrow from $k$ to $i$ in the generated graph and an arc in $H_0$; (5)
the $ki$- and $ij$-edges are the same in $H_0$.

(1) We have that $j\in\an(C)$ in $H$. Hence by Lemma~\ref{lem:22nn}
we have one of the two following cases in $H$: (a) $i$ or a descendant
of $i$ is the endpoint of a line, which, since $H$ is a summary graph,
is a contradiction; (b) $i\in\an(C)$, which by the algorithm implies
that in the generated graph there is no arrowhead at $i$ on the
$ik$-edge, again a contradiction.

(2), (3) We have that $i\in\an(C)$ in $H$, which by the algorithm implies
that in the generated graph there is no arrowhead at $i$ on the
$ik$-edge, a contradiction.

(4), (5) We have that $\langle k,i,j\rangle$ still has an arrowhead
pointing to a line in $H_0$. By Lemma~\ref{lem:22} there is an
arrowhead at $i$ in $H$, hence the line has been generated by the
algorithm. Again by Lemma~\ref{lem:22} we observe that one of the
following holds: (a) $i$ or a descendant of $i$ is the endpoint of a
line, which, since $H$ is a summary graph, is a contradiction; (b) $i\in
\an(C)$, which by the algorithm implies that in the generated graph
there is no arrowhead pointing to $i$ on the $ik$-edge, again a contradiction.

\emph{Part II.} There is also no direction-preserving cycles in the
generated graph: If, for contradiction, there is a direction-preserving
cycle in the generated graph then at least one
arrow, say from $k$ to $l$, should be generated by the algorithm since
there is no direction-preserving cycle in $H$. This arrow can
be generated either by step 1, or by step 2 as an arc replaced by an
arrow. If the $kl$-arrow is generated by step 2 then we have that $k\in
\an(C)$ in $H$. Therefore, there are no arrowheads pointing to $k$ in
the generated graph, which means that $k$ cannot be on a
direction-preserving cycle, a
contradiction.

Therefore, we can assume that the direction-preserving cycle exists in
$H_0$. Since there are no arrowheads pointing to lines in $H$ and $H$
does not contain a direction-preserving cycle, Lemma~\ref{lem:22nn}
implies that a node (and therefore all nodes) of the cycle are in $\an
(C)$. Hence, the arrows turn into lines in the generated graph, a contradiction.
\end{pf*}
\begin{pf*}{Proof of Proposition~\ref{prop:vvn}}
We know that $\mathcal{SG}\subseteq\mathcal{RG}$. Notice that $H$ is
a summary graph. We know that, for the generated directed acyclic graph
$G$, explained in (a), $\alpha_{\mathrm{RG}}(G;M,C)=H$. Suppose, for
contradiction, that step 2 of Algorithm \ref{alg:23} changes the
graph. Thus a node with an arrowhead pointing to is in $\an(C)$, which
implies that it is an ancestor of a node that is the endpoint of a line
or on a direction-preserving cycle in $H$, a contradiction. Therefore,
$\alpha_{\mathrm{SG}}(G;M,C)=H$.
\end{pf*}
\begin{pf*}{Proof of Theorem~\ref{prop:25}}
We show that there is an edge between $i$ and $j$ in $\alpha_{\mathrm{SG}}(H;M\cup M_1,C\cup C_1)$ if and only if there is the same type of edge
in $\alpha_{\mathrm{SG}}(\alpha_{\mathrm{SG}}(H;M,C);M_1,C_1)$. For this purpose for
summary graphs, it is enough to prove that there is an edge between $i$
and $j$ with arrowhead pointing to $j$ in $\alpha_{\mathrm{SG}}(H;M\cup
M_1,C\cup C_1)$ if and only if there is an edge between $i$ and $j$
with arrowhead pointing to $j$
in $\alpha_{\mathrm{SG}}(\alpha_{\mathrm{SG}}(H;M,C);M_1,C_1)$.

($\Rightarrow$) Suppose that in $\alpha_{\mathrm{SG}}(H;M\cup M_1,C\cup C_1)$
there is an $ij$-edge with arrowhead pointing to $j$. We have that
$j\not\in\an(C\cup C_1)$ in $H$ and in $\alpha_{\mathrm{RG}}(H;M\cup
M_1,C\cup C_1)$ there is an $ij$-edge with arrowhead pointing to $j$.
By Lemma~\ref{lem:22}, there is an $m$-connecting path between $i$ and
$j$ given $M\cup M_1$ and $C\cup C_1$ in $H$ with arrowhead pointing to
$j$. By what we showed before in the proof of Theorem~\ref{thm:21n} in
$\alpha_{\mathrm{RG}}(H;M,C)$ there is an $m$-connecting path $\pi$ between
$i$ and $j$ given $M_1$ and $C_1$ with arrowhead pointing to $j$.

Now notice that by step 2 of Algorithm \ref{alg:23} non-collider nodes
remain non-collider. In addition, if a collider {\textsf V} $\langle
h,k,l\rangle$ on $\pi$ turns into non-collider then $k\in\an(C)$
and therefore by step 1 of the algorithm there is an endpoint-identical
$hl$-edge that can be used instead of the {\textsf V} to generate an
$m$-connecting path. Moreover, if the collider node $k$ is in $\an
(C_1)$, and on the direction-preserving path an arrow turns into a line
then $k\in\an(C)$ in $H$ and once again there is an $hl$-edge to be
used instead of the {\textsf V} to establish an $m$-connecting path.
Therefore, there is an $m$-connecting path between $i$ and $j$ given
$M_1$ and $C_1$ in $\alpha_{\mathrm{SG}}(H;M,C)$. We also have that since
$j\not\in\an(C)$ in $H$, there is an arrowhead pointing $j$ on the path.

Now by Lemma~\ref{lem:22} in $\alpha_{\mathrm{RG}}(\alpha_{\mathrm{SG}}(H;M,C);M_1,C_1)$ there is an $ij$-edge with
arrowhead pointing~$j$. In addition, in $\alpha_{\mathrm{SG}}(H;M,C)$, $j\not\in\an(C_1)$: This
is because if, for contradiction, $j\in\an(C_1)$ in $\alpha_{\mathrm{SG}}(H;M,C)$ then, in $\alpha_{\mathrm{RG}}(H;M,C)$, $j\in\an(C\cup C_1)$.
This by Lemma~\ref{lem:22nn} and the fact that $H$ is a summary graph
implies that $j\in\an(C\cup C_1)$ in $H$, a contradiction.

Therefore, since in $\alpha_{\mathrm{SG}}(H;M,C)$, $j\not\in\an(C_1)$, there
is an arrowhead pointing $j$ on the $ij$-edge in $\alpha_{\mathrm{SG}}(\alpha_{\mathrm{SG}}(H;M,C);M_1,C_1)$.

($\Leftarrow$) Suppose that in $\alpha_{\mathrm{SG}}(\alpha_{\mathrm{SG}}(H;M,C);M_1,C_1)$ there is an $ij$-edge with arrowhead pointing to
$j$. This implies that $j\not\in\an(C_1)$ in $\alpha_{\mathrm{SG}}(H;M,C)$.
In $\alpha_{\mathrm{RG}}(\alpha_{\mathrm{SG}}(H;M,C);\break M_1,C_1)$ there is also an
$ij$-edge with arrowhead pointing to $j$. By Lemma~\ref{lem:22}, in
$\alpha_{\mathrm{SG}}(H;M,C)$ there is an $m$-connecting path $\pi$ given
$M_1$ and $C_1$ between $i$ an $j$ with arrowhead pointing to $j$ on
the path. This implies that $j\not\in\an(C)$ in $H$.

By step 2 of the algorithm, collider nodes on $\pi$ in $\alpha_{\mathrm{SG}}(H;M,C)$ are colliders in $\alpha_{\mathrm{RG}}(H;\allowbreak M,C)$. In addition, if a
non-collider {\textsf V} $\langle h,k,l\rangle$ on $\pi$ is collider in
$\alpha_{\mathrm{RG}}(H;M,C)$ then $k\in\an(C)$ in $H$ and therefore by step
1 of the algorithm there is an endpoint-identical $hl$-edge that can be
used instead of the {\textsf V} to generate an $m$-connecting path in
$\alpha_{\mathrm{RG}}(H;M,C)$. Moreover, if the collider node $k$ on $\pi$ is
in $\an(C_1)$, and on the direction-preserving path an arrow is an arc
in $\alpha_{\mathrm{RG}}(H;M,C)$ then the collider node is in $\an(C)$ in $H$
and once again there is an $hl$-edge to be used instead of the {\textsf V}
to establish an $m$-connecting path. Therefore, there is an
$m$-connecting path between $i$ and $j$ given $M_1$ and $C_1$ in
$\alpha_{\mathrm{RG}}(H;M,C)$ with arrowhead pointing to $j$ on the path. By
what we showed before in the proof of Theorem~\ref{thm:21n} in $H$
there is an $m$-connecting path between $i$ and $j$ given $M\cup M_1$
and $C\cup C_1$ with arrowhead pointing to $j$ on the path.

Lemma~\ref{lem:22} implies that in $\alpha_{\mathrm{RG}}(H;M\cup M_1,C\cup
C_1)$ there is an $ij$-edge with arrowhead pointing to $j$. We showed
before that $j\not\in\an(C)$ in $H$. In addition, in $H$, $j\not\in
\an(C_1)$: Suppose, for contradiction, that $j\in\an(C_1)$ in $H$.
This direction-preserving path is also direction-preserving in $\alpha_{\mathrm{RG}}(H;M,C)$ unless possibly a node $t$ on the path is in $M$ or in
$C$. In the former case one can skip $t$ and obtain a
direction-preserving path. In the latter case $j\in\an(C)$ in $H$,
which is not permissible. Therefore, in $\alpha_{\mathrm{RG}}(H;M,C)$, $j\in
\an(C_1)$. Hence, since $j\not\in\an(C)$ in $H$, $j\in\an(C_1)$
in $\alpha_{\mathrm{SG}}(H;M,C)$, a contradiction.

Therefore, $j\not\in\an(C\cup C_1)$ in $H$, and the $ij$-edge has
arrowhead pointing to $j$ in $\alpha_{\mathrm{SG}}(H;M\cup M_1,C\cup C_1)$.
\end{pf*}
\begin{pf*}{Proof of Theorem~\ref{thm:22}}
By what we proved in Theorem~\ref{thm:21} it is enough to
show between $i$ and $j$ there is an $m$-connecting path given $C_1$ in
$\alpha_{\mathrm{RG}}(H;M,C)$ if and only if there is an $m$-connecting path
given $C_1$
in $\alpha_{\mathrm{SG}}(H;M,C)$.

($\Rightarrow$) Consider an $m$-connecting path given $C_1$ in $\alpha_{\mathrm{RG}}(H;M,C)$ between $i$ and $j$. There obviously exists an
$m$-connecting path $\pi$ given $C_1$ in $\alpha_{\mathrm{RG}}(H;M,C)$.
Now by step 2 of the algorithm non-collider nodes remain non-collider.
In addition, if a collider {\textsf V} $\langle h,k,l\rangle$ on $\pi$
turns into non-collider then $k\in\an(C)$ and therefore by step 1 of
the algorithm there is an endpoint-identical $hl$-edge generated that
can be used instead of the {\textsf V} to generate an $m$-connecting path.
Moreover, if the collider node $k$ is in $\an(C_1)$, and on the
direction-preserving path an arrow turns into a line then $k\in\an
(C)$ in $H$ and once again there is an $hl$-edge to be used instead of
the {\textsf V} to establish an $m$-connecting path. Therefore, there is an
$m$-connecting path between $i$ and $j$ given $C_1$ in $\alpha_{\mathrm{SG}}(H;M,C)$.

($\Leftarrow$) Consider an $m$-connecting path $\pi$ given $C_1$ in
$\alpha_{\mathrm{SG}}(H;M,C)$ between $i$ and $j$. By step 2 of the algorithm,
collider nodes on $\pi$ in $\alpha_{\mathrm{SG}}(H;M,C)$ are colliders in
$\alpha_{\mathrm{RG}}(H;M,C)$. In addition, if a non-collider {\textsf V} $\langle
h,k,l\rangle$ on $\pi$ is collider in $\alpha_{\mathrm{RG}}(H;M,C)$ then
$k\in\an(C)$ in $H$ and therefore by step 1 of the algorithm there is
an endpoint-identical $hl$-edge that can be used instead of the {\textsf V}
to generate an $m$-connecting path in $\alpha_{\mathrm{RG}}(H;M,C)$. Moreover,
if the collider node $k$ on $\pi$ is in $\an(C_1)$, and on the
direction-preserving path an arrow is an arc in $\alpha_{\mathrm{RG}}(H;M,C)$
then $k\in\an(C)$ in $H$ and once again there is an $hl$-edge to be
used instead of the {\textsf V} to establish an $m$-connecting path.
Therefore, there is an $m$-connecting path between $i$ and $j$ given
$C_1$ in $\alpha_{\mathrm{RG}}(H;M,C)$.
\end{pf*}
\begin{pf*}{Proof of Proposition~\ref{prop:26}}
To prove that the graph generated by Algorithm \ref{alg:25} is an AG,
first notice that graphs generated by step 1 of Algorithm \ref{alg:25}
are summary graphs. Therefore, since steps 2 and 3 do not generate any
lines, it is enough to prove that steps 2 and 3 of the algorithm remove
all subgraphs where there is an arc with one endpoint that is an
ancestor of the other endpoint in the generated summary graph, and that
these do not generate any direction-preserving cycles.

Step 3 of the algorithm removes all such subgraphs. Step 2 does not
generate any direction-preserving cycles by adding an arrow to the
graph: Consider the first iteration of the algorithm, where, for
contradiction, a direction-preserving cycle is generated. If it is
generated by generating an arrow from $i$ to $j$ then we know that
there is $i\fra k\arc j$, where $k\in\an(j)$. Denote the
direction-preserving path from $k$ to $j$ by $\pi_1$ and the
direction-preserving path from $j$ to $i$ which, together with the
generated $ij$-arrow, establishes a direction-preserving cycle by $\pi_2$. It is seen that in the previous iteration of the algorithm
$\langle\pi_2,k,\pi_1\rangle$ is a direction-preserving cycle, a
contradiction.

Step 3 does not generate any direction-preserving cycles by replacing
an arc by an arrow: Consider the first iteration of the algorithm,
where, for contradiction, a direction-preserving cycle is generated. If
it is generated by replacing an $ij$-arc by an arrow from $i$ to $j$,
then we know there is a direction-preserving path $\pi_1$ from $i$ to
$j$. Denote the direction-preserving path from $j$ to $i$ which,
together with the generated $ij$-arrow, establishes a
direction-preserving cycle by $\pi_2$. It is seen that in the
previous iteration of the algorithm $\langle\pi_1,\pi_2\rangle$ is
a direction-preserving cycle, a contradiction.
\end{pf*}
We use the following lemma to prove Theorem~\ref{thm:25}. For more
descriptive proofs for the following results, see \citep{sad12}.
%
\begin{lemma}\label{lem:2nn}
Let $H$ be a summary graph and $M$ and $C$ be two subsets of its node
set. It holds that $\alpha_{\mathrm{AG}}(\alpha_{\mathrm{SG}.\mathrm{AG}}(H);M,C)=\alpha_{\mathrm{AG}}(H;M,C)$.
\end{lemma}
\begin{pf}
There are two differences between $H$ and $\alpha_{\mathrm{SG}.\mathrm{AG}}(H)$: (1) For
an $ij$-arc such that $j\in\an(i)$ in $H$, there is an arrow from $j$
to $i$ replaced in $\alpha_{\mathrm{SG}.\mathrm{AG}}(H)$; (2) for a primitive inducing
path $\pi$ between $i$ and $j$ in $H$, there is an endpoint-identical
$ij$-edge in $\alpha_{\mathrm{SG}.\mathrm{AG}}(H)$.

Notice that $\operatorname{an}(C)$ is the same in both graphs. After applying step 1
of Algorithm \ref{alg:23} (a part of step 1 of Algorithm \ref
{alg:25}), for each difference, the following occurs:

(1) This step of the algorithm may generate further differences for (1)
if there is a $kj$-edge with an arrowhead pointing to $j$ and $j\in
C\cup\an(C)$. In this case, there is a $ki$-edge in $H$. However,
such an edge already exists in $\alpha_{\mathrm{SG}.\mathrm{AG}}(H)$ since $\langle
k,j,i\rangle$ in $H$ generates an edge by $\alpha_{\mathrm{SG}.\mathrm{AG}}$;

(2) This step of the algorithm may generate further differences for (2)
if $\pi$ has more than three nodes and one of the two following cases
occurs: (a) there is an arrowhead pointing to $j$ on $\pi$, there is a
$kj$-edge with an arrowhead pointing to $j$, and $j\in C\cup\an(C)$;
(b) there is a $kj$-edge with no arrowhead pointing to $j$, and $j\in
M$. In both cases, by using the $\langle i,j,k\rangle$-{\textsf V} a
$ki$-edge is generated in $\alpha_{\mathrm{SG}.\mathrm{AG}}(H)$. In $H$, by using the
$\langle h,j,k\rangle$-{\textsf V}, where $h$ is the node adjacent to $j$
on $\pi$, a $kh$-edge is generated, which establishes an
endpoint-identical primitive inducing path between $i$ and $k$ in
$\alpha_{\mathrm{SG}.\mathrm{AG}}(H)$;
%

After applying step 2 of Algorithm \ref{alg:23}, the following occurs:
(1) This step of the algorithm may generate further differences for (1)
if $j\in C\cup\an(C)$. In this case, $ij$-arc in $H$ turns into an
arrow from $j$ to $i$, which, however, already exists in $\alpha_{\mathrm{SG}.\mathrm{AG}}(H)$;
(2) This step of the algorithm may generate further differences for (2)
if any of the nodes on $\pi$, say $l$, is in $C\cup\an(C)$. However,
in $H$, by the previous step of the algorithm, the nodes adjacent to
$l$ on $\pi$ have become adjacent, and established a shorter primitive
inducing path between $i$ and $k$ (or $j$).

Hence, thus far, the differences between the two generated graphs are
the same as the differences between $H$ and $\alpha_{\mathrm{SG}.\mathrm{AG}}(H)$.
Therefore, by applying steps 2 and 3 of Algorithm \ref{alg:25}, the
same graphs will be generated.
\end{pf}
\begin{pf*}{Proof of Theorem~\ref{thm:25}}
Using Theorem~\ref{prop:25}, Proposition~\ref{prop:35}, and Lemma~\ref{lem:2nn}, we have the following:
\begin{eqnarray*}
&&\alpha_{\mathrm{AG}} \bigl(\alpha_{\mathrm{AG}}(H;M,C);M_1,C_1
\bigr)\\
&& \quad =\alpha_{\mathrm{AG}} \bigl(\alpha_{\mathrm{SG}.\mathrm{AG}} \bigl(
\alpha_{\mathrm{SG}}(H;M,C) \bigr);M_1,C_1 \bigr)
 =\alpha_{\mathrm{AG}} \bigl(\alpha_{\mathrm{SG}}(H;M,C);M_1,C_1
\bigr)\\
&& \quad = \alpha_{\mathrm{SG}.\mathrm{AG}} \bigl(\alpha_{\mathrm{SG}} \bigl(
\alpha_{\mathrm{SG}}(H;M,C);M_1,C_1 \bigr) \bigr)
  =\alpha_{\mathrm{SG}.\mathrm{AG}} \bigl(\alpha_{\mathrm{SG}}(H;M\cup M_1,C\cup
C_1) \bigr)\\
&& \quad =\alpha_{\mathrm{AG}}(H;M\cup M_1,C\cup
C_1).
\end{eqnarray*}
\upqed
\end{pf*}
\begin{pf*}{Proof of Theorem~\ref{thm:23}}
By what we proved in Theorem~\ref{thm:22}, it is enough to
show between $A$ and $B$ there is an $m$-connecting path given $C_1$ in
$\alpha_{\mathrm{SG}}(H;M,C)$ if and only if there is an $m$-connecting path
given $C_1$
in $\alpha_{\mathrm{AG}}(H;M,C)$.

Let $\langle\alpha_{\mathrm{SG}}(H;M,C)=H_0,H_1,\ldots,H_m\rangle$ be
intermediate graphs that have each been generated by adding one edge to the
previous graph by step 2 of Algorithm \ref{alg:25}. In addition, let
$\langle H_m=H'_0,H'_1,\ldots,H'_n=\alpha_{\mathrm{AG}}(H;M,C)\rangle$ be
intermediate graphs that have each been generated by replacing one edge
in the previous graph by step 3 of Algorithm \ref{alg:25}.

Suppose that in the step between $H_p$ and $H_{p+1}$ an arrow from $j$
to $i$ or an arc between $i$ and $j$ for {\textsf V} $j\fra k\arc i$ or
{\textsf V} $j\arc k\arc i$, when $k\in\an(i)$, is generated. It holds
that there is an $m$-connecting path between $A$ and $B$ given
$C_1$ in $H_p$ if and only if there is an $m$-connecting path between
$A$ and $B$ given $C_1$ in $H_{p+1}$. This is because if $k\in
C_1$ or one of the descendants of $k$ on the direction-preserving path
from $k$ to $i$ is in $C_1$ then the $ij$-edge and the {\textsf V} $j\arc
k\arc i$ can be interchanged on the $m$-connecting path. If these nodes
are not in $C_1$ then the $ij$-edge and the path made up by the
$jk$-edge and the direction-preserving path from $k$ to $i$ can be
interchanged. By induction in one direction and reverse induction in
the other direction, we conclude that there is an $m$-connecting path
given $C_1$ in $\alpha_{\mathrm{SG}}(H;M,C)$ if and only if there is an
$m$-connecting path given $C_1$ in~$H_m$.

Now suppose that in the step between $H'_{p'}$ and $H'_{p'+1}$ an arrow
from $j$ to $i$ has\vspace*{1pt} been replaced by an arc between $j$ and $i$, where
$j\in\an(i)$. The only interesting case here is when there is an edge
between $j$ and another node $l$ with arrowhead pointing to $j$. In
this case, an edge has been already generated between $l$ and $i$ by
step 2 of the algorithm. Therefore, again by induction in one direction
and reverse induction in the other direction the result follows.
\end{pf*}
\end{appendix}

\section*{Acknowledgements}
The author is grateful to Steffen Lauritzen, Nanny Wermuth, and the
referees for helpful comments.


%

\printhistory


\begin{thebibliography}{18}

\bibitem{cox93}
%
\begin{barticle}[mr]
\bauthor{\bsnm{Cox},~\bfnm{D.~R.}\binits{D.R.}} \AND
\bauthor{\bsnm{Wermuth},~\bfnm{Nanny}\binits{N.}}
(\byear{1993}).
\btitle{Linear dependencies represented by chain graphs}.
\bjournal{Statist. Sci.}
\bvolume{8}
\bpages{204--218, 247--283}.
\bnote{With comments and a rejoinder by the authors}.
\bid{issn={0883-4237}, mr={1243593}}
\bptnote{check related}
\bptok{imsref}%
\end{barticle}
%
\endbibitem

\bibitem{cox96}
%
\begin{bbook}[mr]
\bauthor{\bsnm{Cox},~\bfnm{D.~R.}\binits{D.R.}} \AND
\bauthor{\bsnm{Wermuth},~\bfnm{Nanny}\binits{N.}}
(\byear{1996}).
\btitle{Multivariate Dependencies: Models, Analysis and Interpretation}.
\bseries{Monographs on Statistics and Applied Probability}
\bvolume{67}.
\baddress{London}: \bpublisher{Chapman \& Hall}.
\bid{mr={1456990}}
\bptok{imsref}%
\end{bbook}
%
\endbibitem

\bibitem{drt04}
%
\begin{binproceedings}[author]
\bauthor{\bsnm{Drton},~\bfnm{M.}\binits{M.}} \AND
\bauthor{\bsnm{Richardson},~\bfnm{T.}\binits{T.}}
(\byear{2004}).
\btitle{Iterative Conditional Fitting for Gaussian Ancestral Graph Models}.
In \bbooktitle{Proceedings of the Proceedings of the Twentieth Conference
Annual Conference on Uncertainty in Artificial Intelligence (UAI-04)}
\bpages{130--137}.
\baddress{Arlington, VA}: \bpublisher{AUAI Press}.
\bptok{imsref}%
\end{binproceedings}
%
\endbibitem

\bibitem{kii84}
%
\begin{barticle}[mr]
\bauthor{\bsnm{Kiiveri},~\bfnm{Harri}\binits{H.}},
\bauthor{\bsnm{Speed},~\bfnm{T.~P.}\binits{T.P.}} \AND
\bauthor{\bsnm{Carlin},~\bfnm{J.~B.}\binits{J.B.}}
(\byear{1984}).
\btitle{Recursive causal models}.
\bjournal{J. Austral. Math. Soc. Ser. A}
\bvolume{36}
\bpages{30--52}.
\bid{issn={0263-6115}, mr={0719999}}
\bptok{imsref}%
\end{barticle}
%
\endbibitem

\bibitem{kos02}
%
\begin{barticle}[mr]
\bauthor{\bsnm{Koster},~\bfnm{Jan T.~A.}\binits{J.T.A.}}
(\byear{2002}).
\btitle{Marginalizing and conditioning in graphical models}.
\bjournal{Bernoulli}
\bvolume{8}
\bpages{817--840}.
\bid{issn={1350-7265}, mr={1963663}}
\bptok{imsref}%
\end{barticle}
%
\endbibitem

\bibitem{lau96}
%
\begin{bbook}[mr]
\bauthor{\bsnm{Lauritzen},~\bfnm{Steffen~L.}\binits{S.L.}}
(\byear{1996}).
\btitle{Graphical Models}.
\bseries{Oxford Statistical Science Series}
\bvolume{17}.
\baddress{New York}: \bpublisher{Clarendon Press}.
\bid{mr={1419991}}
\bptok{imsref}%
\end{bbook}
%
\endbibitem

\bibitem{lau90}
%
\begin{barticle}[mr]
\bauthor{\bsnm{Lauritzen},~\bfnm{S.~L.}\binits{S.L.}},
\bauthor{\bsnm{Dawid},~\bfnm{A.~P.}\binits{A.P.}},
\bauthor{\bsnm{Larsen},~\bfnm{B.~N.}\binits{B.N.}} \AND
\bauthor{\bsnm{Leimer},~\bfnm{H.~G.}\binits{H.G.}}
(\byear{1990}).
\btitle{Independence properties of directed {M}arkov fields}.
\bjournal{Networks}
\bvolume{20}
\bpages{491--505}.
\bid{doi={10.1002/net.3230200503}, issn={0028-3045}, mr={1064735}}
\bptok{imsref}%
\end{barticle}
%
\endbibitem

\bibitem{pea88}
%
\begin{bbook}[mr]
\bauthor{\bsnm{Pearl},~\bfnm{Judea}\binits{J.}}
(\byear{1988}).
\btitle{Probabilistic Reasoning in Intelligent Systems: Networks of Plausible
Inference}.
\bseries{The Morgan Kaufmann Series in Representation and Reasoning}.
\baddress{San Mateo, CA}: \bpublisher{Morgan Kaufmann}.
\bid{mr={0965765}}
\bptok{imsref}%
\end{bbook}
%
\endbibitem

\bibitem{pea94}
%
\begin{barticle}[author]
\bauthor{\bsnm{Pearl},~\bfnm{J.}\binits{J.}} \AND
\bauthor{\bsnm{Wermuth},~\bfnm{N.}\binits{N.}}
(\byear{1994}).
\btitle{When can association graphs admit a causal interpretation?}
\bjournal{Models and Data, Artificial Intelligence and Statistics}
\bvolume{4}
\bpages{205--214}.
\bptok{imsref}%
\end{barticle}
%
\endbibitem

\bibitem{ric02}
%
\begin{barticle}[mr]
\bauthor{\bsnm{Richardson},~\bfnm{Thomas}\binits{T.}} \AND
\bauthor{\bsnm{Spirtes},~\bfnm{Peter}\binits{P.}}
(\byear{2002}).
\btitle{Ancestral graph {M}arkov models}.
\bjournal{Ann. Statist.}
\bvolume{30}
\bpages{962--1030}.
\bid{doi={10.1214/aos/1031689015}, issn={0090-5364}, mr={1926166}}
\bptok{imsref}%
\end{barticle}
%
\endbibitem

\bibitem{sad12}
%
\begin{bphdthesis}[author]
\bauthor{\bsnm{Sadeghi},~\bfnm{Kayvan}\binits{K.}}
(\byear{2012}).
\btitle{Graphical representation of independence structures}.
\btype{Ph.D. thesis}, \bschool{Univ. Oxford}.
\bptok{imsref}%
\end{bphdthesis}
%
\endbibitem

\bibitem{stu05}
%
\begin{bbook}[author]
\bauthor{\bsnm{Studeny},~\bfnm{M.}\binits{M.}}
(\byear{2005}).
\btitle{Probabilistic Conditional Independence Structures}.
\baddress{London, United Kingdom}: \bpublisher{Springer}.
\bptok{imsref}%
\end{bbook}
%
\endbibitem

\bibitem{wer08}
%
\begin{barticle}[mr]
\bauthor{\bsnm{Wermuth},~\bfnm{Nanny}\binits{N.}}
(\byear{2011}).
\btitle{Probability distributions with summary graph structure}.
\bjournal{Bernoulli}
\bvolume{17}
\bpages{845--879}.
\bid{doi={10.3150/10-BEJ309}, issn={1350-7265}, mr={2817608}}
\bptok{imsref}%
\end{barticle}
%
\endbibitem

\bibitem{werc04}
%
\begin{barticle}[mr]
\bauthor{\bsnm{Wermuth},~\bfnm{Nanny}\binits{N.}} \AND
\bauthor{\bsnm{Cox},~\bfnm{D.~R.}\binits{D.R.}}
(\byear{2004}).
\btitle{Joint response graphs and separation induced by triangular systems}.
\bjournal{J. R. Stat. Soc. Ser. B Stat. Methodol.}
\bvolume{66}
\bpages{687--717}.
\bid{doi={10.1111/j.1467-9868.2004.b5161.x}, issn={1369-7412}, mr={2088296}}
\bptok{imsref}%
\end{barticle}
%
\endbibitem

\bibitem{werc08}
%
\begin{barticle}[mr]
\bauthor{\bsnm{Wermuth},~\bfnm{Nanny}\binits{N.}} \AND
\bauthor{\bsnm{Cox},~\bfnm{D.~R.}\binits{D.R.}}
(\byear{2008}).
\btitle{Distortion of effects caused by indirect confounding}.
\bjournal{Biometrika}
\bvolume{95}
\bpages{17--33}.
\bid{doi={10.1093/biomet/asm092}, issn={0006-3444}, mr={2409712}}
\bptok{imsref}%
\end{barticle}
%
\endbibitem

\bibitem{wer94}
%
\begin{btechreport}[author]
\bauthor{\bsnm{Wermuth},~\bfnm{N.}\binits{N.}},
\bauthor{\bsnm{Cox},~\bfnm{D.~R.}\binits{D.R.}} \AND
\bauthor{\bsnm{Pearl},~\bfnm{J.}\binits{J.}}
(\byear{1994}).
\btitle{Explanation for multivariate structures derived from univariate
recursive regressions}.
\btype{Technical Report}   \bnumber{94(1)},
\binstitution{Univ. Mainz}, \baddress{Germany}.
\bptok{imsref}%
\end{btechreport}
%
\endbibitem

\bibitem{wer06}
%
\begin{barticle}[mr]
\bauthor{\bsnm{Wermuth},~\bfnm{N.}\binits{N.}},
\bauthor{\bsnm{Wiedenbeck},~\bfnm{M.}\binits{M.}} \AND
\bauthor{\bsnm{Cox},~\bfnm{D.~R.}\binits{D.R.}}
(\byear{2006}).
\btitle{Partial inversion for linear systems and partial closure of
independence graphs}.
\bjournal{BIT}
\bvolume{46}
\bpages{883--901}.
\bid{doi={10.1007/s10543-006-0093-9}, issn={0006-3835}, mr={2285213}}
\bptok{imsref}%
\end{barticle}
%
\endbibitem

\end{thebibliography}
\end{document}